\documentclass[11pt,a4paper]{article}

\usepackage[colorlinks,pdfstartview=FitH,citecolor=blue,linkcolor=blue,urlcolor=blue]{hyperref}

\usepackage{graphicx}
\usepackage{cite}

\usepackage{pgfplots}                      
\pgfplotsset{compat=1.18}       
\usepgfplotslibrary{fillbetween}           

\usepackage{tikz}
\usetikzlibrary{matrix,decorations.pathreplacing,positioning,calc, patterns,backgrounds, intersections}
\usepackage{hyperref}
\usepackage{amsmath,amsthm,amssymb,bm,relsize}
\usepackage{calrsfs}
\usepackage{scalerel}
\usepackage{verbatim}

\usetikzlibrary{decorations.markings}

\usepackage{float}

\usepackage{xcolor}
\newlength{\mynodespace}
\pgfplotsset{compat=1.12}

\usepackage{setspace}
\usepackage{multirow}
\usepackage{array}
\newcolumntype{x}[1]{>{\centering\arraybackslash\hspace{0pt}}p{#1}}

\usepackage{empheq}
\usepackage{wasysym}

\usepackage{caption}
\usepackage{subcaption}

\usetikzlibrary{arrows.meta}

\usetikzlibrary{knots}
\newcommand{\wt}{\textnormal{wt}}
\usetikzlibrary{spath3}
\usetikzlibrary{hobby}

\usepackage{enumitem}
\setitemize{itemsep=-1pt}
\setenumerate{itemsep=-1pt}

\usepackage{titling}
\usepackage[margin=2.8cm]{geometry}
\usepackage{pgfplots}
\pgfplotsset{width=10cm,compat=1.9}

\definecolor{myg}{RGB}{220,220,220}

\usepackage[capitalise]{cleveref}

\theoremstyle{definition}

\newtheorem{theorem}{Theorem}

\newtheorem{proposition}[theorem]{Proposition}
\newtheorem{lemma}[theorem]{Lemma}

\newtheorem{definition}[theorem]{Definition}
\newtheorem{example}[theorem]{Example}

\newtheorem{remark}[theorem]{Remark}

\newtheorem*{theorem*}{Theorem}

\newcommand{\numberset}{\mathbb}

\newcommand{\F}{\numberset{F}}
\newcommand{\R}{\numberset{R}}
\newcommand{\Z}{\numberset{Z}}
\newcommand{\mA}{\mathcal{A}}
\newcommand{\mC}{\mathcal{C}}
\newcommand{\mD}{\mathcal{D}}

\newcommand{\Aut}{\textnormal{Aut}}

\title{\textbf{A New Invariant for Prime Alternating Knots From Error-Correcting Codes}}

\usepackage{authblk}

\author{Altan B. K\i l\i\c{c}, Ruud Pellikaan, Alberto Ravagnani \\ 
Department of Mathematics and Computer Science, Eindhoven University of Technology, the Netherlands
\thanks{A. B. K\i l\i\c{c} is supported by the Dutch Research Council through grant VI.Vidi.203.045. and A. Ravagnani is supported by the Dutch Research Council through grants VI.Vidi.203.045, 
OCENW.KLEIN.539, 
and by the Royal Academy of Arts and Sciences of the Netherlands.}
\thanks{Emails: a.b.kilic@tue.nl, g.r.pellikaan@outlook.com, a.ravagnani@tue.nl.}}

\date{}

\begin{document}

\maketitle

\vspace{ 1 cm}
\begin{abstract}
This paper shows that the Alexander-Briggs code of a knot gives rise to a new invariant that distinguishes prime alternating knots. The restriction to prime alternating knots precisely follows from the fact that our approach relies on Tait's flyping theorem. We also provide examples where the new invariant succeeds in separating knots that the well-known invariants, such as some knot polynomials, fail.  
\end{abstract}

\medskip

\section{Introduction}
\label{sec:intro}

One of the most fundamental problems in knot theory is to determine whether two knots are equivalent \cite{kawauchi1996survey, crowell2012introduction}. Knot invariants are very powerful for tackling the ``negative problem'' of knot equivalence. Among the many classes, alternating knots form one of the most central and best understood families in knot theory that also admit associated graph constructions, which creates a connection between knot theory and combinatorial structures \cite{thistlethwaite1987}.

Although the very idea of a knot invariant is rooted in knot theory itself, the tools used to construct these invariants often come from other areas of mathematics, for example, algebra via polynomials and knot groups, see \cite{alexander1928topological,fox1953free,jones1985,jones1986} among others. In this paper, we consider a discipline that has only been explored in \cite{kilic2024knot} in this regard and provide an invariant using the main object of that field, namely, coding theory and error-correcting codes \cite{macwilliams1977theory,huffman2010fundamentals}. This provides further evidence that connecting coding theory to other areas of mathematics often leads to new discoveries and applications.

We adopt the setting of \cite{kilic2024knot} where a bridge
between coding theory and mathematical knot theory is constructed with a combinatorial flavor. In \cite{kilic2024knot}, error correcting codes that are coming from the colorings of knot and link diagrams are introduced, with a main focus on the colorings of the strands and regions, called Fox and Dehn colorings respectively, where the colors have values in a ring. In this paper, we consider Alexander-Briggs (AB) colorings, the colorings of the crossings of a knot diagram, and the corresponding AB code.

The main result of the paper shows that the ($\pm1$)-permutation equivalence class of the Alexander-Briggs code of a knot gives a new invariant that distinguishes prime alternating knots, and is stated as follows.

\begin{theorem}
\label{cor-code-alternating}
The ($\pm1$)-permutation equivalence class of the Alexander-Briggs code $\mA_{D}$ of a reduced Tait diagram $D$ of a prime alternating knot is 
an invariant of the knot.
\end{theorem}

\noindent \textbf{Outline of the paper:} Section \ref{sec:pre} is devoted for the background needed in the paper, and it consists of two subsections, where the first is about Fox, Dehn and Alexander-Briggs colorings, and the second focuses on operations in graphs and codes we use throughout the paper. Section \ref{sec:Tait} discusses Tait's flyping conjecture and places it in the right context for our purposes. Section \ref{sec:newinvariant} contributes to knot theory by introducing a new invariant for prime alternating knots arising from coding theory. In Section \ref{sec:comparenewinv}, we compare our invariant with already established, well-known invariants and provide examples where it outperforms them. Finally, in Section \ref{sec:mutant}, we take a look at mutant knots and show that our invariant does not distinguish them.

\section{Preliminaries}
\label{sec:pre}
In this section, we recall preliminary results on knot colorings, and review certain operations on graphs and codes that are needed for the rest of the paper. We also establish the notation. 

We start with the definition of a knot and its diagram. (For the sake of conciseness, we refer to \cite{lickorish1997}, among many others, for the definitions of links and their equivalence.)

\begin{definition}
\label{def:knot}
A \textbf{knot} $K$ is a topological subspace of $\R^3$ (or $S^3$, see Remark \ref{rem:whyS3}) that is homeomorphic to the unit circle $S^1$. A \textbf{knot diagram} is a regular projection of a knot onto a plane, drawn in such a way (with little gaps) that one can distinguish whether the
knot passes over or under itself, see Figure \ref{fig:fig8}. Coloring the regions of a knot diagram with two colors, black and white, in a way that two regions have the same color at each crossing if and only if they are not adjacent, called the \textbf{checkerboard coloring}. 
\end{definition}

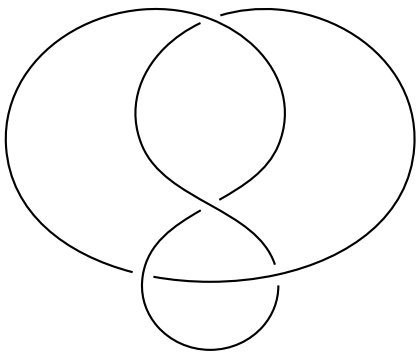
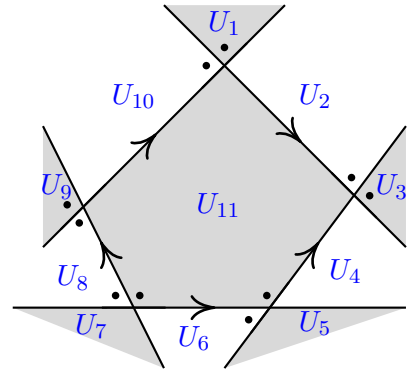
\begin{figure}[h!]
    \centering
\begin{subfigure}{.4\textwidth}
\centering

 \begin{tikzpicture}[scale= 1.8, use Hobby shortcut]
    \path[spath/save=figure8]
    ([closed]0,0) .. (1.5,1) .. (.5,2) ..
    (-.5,1) .. (.5,0) .. (0,-.5) .. (-.5,0) ..
    (.5,1) .. (-.5,2) .. (-1.5,1) .. (0,0);
    \tikzset{
    every spath component/.style={thick, draw},
    spath/knot={figure8}{8pt}{1,3,...,7}
    }
    \path (0,-.7);
    \end{tikzpicture}
    \caption{A figure-eight knot.}
    \label{fig:fig8}

\end{subfigure}
\hspace{0.1\textwidth}
\begin{subfigure}{.4\textwidth}
\centering
\begin{tikzpicture}[scale=0.8, pics/arrow/.style={code={%
  \draw[line width=0pt,{Computer Modern Rightarrow[line
  width=1pt,width=2ex,length=2ex]}-] (0.5ex,0) -- (-0.5ex,0);
  }},pics/rrarrow/.style={code={%
  \draw[line width=0pt,{Computer Modern Rightarrow[line
  width=1pt,width=2ex,length=2ex]}-] (-0.5ex,0) -- (0.5ex,0);
  }}]

\fill[gray!30] (0,2) -- (1,3) -- (-1,3)  -- cycle;
\fill[gray!30] (-7/3,-1/3) -- (-3,1) -- (-3,-1)  -- cycle;

\fill[gray!30] (-3.5,-2) -- (-3/2,-2) -- (-1,-3)  -- cycle;
\fill[gray!30] (0,-3) -- (3/4,-2) -- (3,-2) -- cycle;

\fill[gray!30] (15/7,-1/7) -- (3,1) -- (3,-1)  -- cycle;
\fill[gray!30] (0,2) -- (-7/3,-1/3) -- (-3/2,-2) -- (3/4,-2) -- (15/7,-1/7) -- cycle;

\begin{knot}[
clip width=0,
]
\strand[black,thick] (-3,-1) to pic[pos=0.45,sloped]{arrow} (1,3);
\strand[black,thick] (-1,3) to pic[pos=0.55,sloped]{arrow} (3,-1);
\strand[black,thick] (3,1) to pic[pos=0.5,sloped]{arrow} (0,-3);
\strand[black,thick] (3,-2) to pic[pos=0.5,sloped]{arrow} (-3.5,-2);
\strand[black,thick] (-1,-3) to pic[pos=0.5,sloped]{rrarrow} (-3,1);
\end{knot}

\tikzset{nnode/.style = {shape=circle,fill=myg,draw,inner sep=0pt,minimum
size=1.9em}}

\node[text=blue] at (-0.1,-0.3) {$U_{11}$};
\node[text=blue] at (0,2.7) {$U_1$};
\node[text=blue] at (1.5,1.5) {$U_2$};
\node[text=blue] at (2.77,0) {$U_3$};
\node[text=blue] at (2,-1.4) {$U_4$};
\node[text=blue] at (1.5,-2.25) {$U_5$};
\node[text=blue] at (-0.5,-2.5) {$U_6$};
\node[text=blue] at (-2.2,-2.3) {$U_7$};
\node[text=blue] at (-2.5,-1.5) {$U_8$};
\node[text=blue] at (-2.77,-0) {$U_9$};
\node[text=blue] at (-1.5,1.5) {$U_{10}$};

\node at (0,2.3) [circle, fill=black, inner sep=1pt] {};
\node at (-0.3,2) [circle, fill=black, inner sep=1pt] {};
\node at (2.4,-0.15) [circle, fill=black, inner sep=1pt] {};
\node at (2.1,0.15) [circle, fill=black, inner sep=1pt] {};
\node at (0.4,-2.2) [circle, fill=black, inner sep=1pt] {};
\node at (0.7,-1.8) [circle, fill=black, inner sep=1pt] {};
\node at (-1.8,-1.8) [circle, fill=black, inner sep=1pt] {};
\node at (-1.4,-1.8) [circle, fill=black, inner sep=1pt] {};
\node at (-2.6,-0.3) [circle, fill=black, inner sep=1pt] {};
\node at (-2.4,-0.6) [circle, fill=black, inner sep=1pt] {};

\end{tikzpicture}
    \caption{A Tait diagram of a knot with a checkerboard coloring.}
    \label{fig:Tait}
\end{subfigure}

\caption{A knot diagram and a Tait diagram.}
\label{fig:somediagrams}
\end{figure}

Reidemeister introduced three moves within knot diagrams \cite{reidemeister1927} that can be used to show their equivalence. 

\begin{definition}
\label{def:Reide_equiv}
Two diagrams $D$ and $D'$ are called \textbf{equivalent} if $D$ can be transformed into~$D'$ by using a finite sequence of Reidemeister moves. We denote this by~$D \approx D'$.
\end{definition}

Another way to distinguish whether the strand is an overstrand or an understrand is to use \textbf{Tait diagrams} where the gaps in the standard knot diagrams do not exist, but one uses two dots at each crossing of the knot (with a chosen orientation) placed next to the left 
hand side of an overstrand such that the first dot is placed just before and the second is placed just after the understrand. See Figure \ref{fig:Tait} for an illustration.

Throughout this chapter, all knot diagrams will be regular, reduced and alternating unless stated otherwise. See \cite[Definition 2.6, Definition 2.14]{kilic2024knot} for formal definitions. Moreover,~$R$ will always be a commutative ring with unity, and $t$ is a fixed invertable element of that ring.

Note that in Section \ref{sec:Tait} and Section \ref{sec:newinvariant}, knots will be considered in $S^3$, instead of $R^3$ as in Definition \ref{def:knot}, which is the one-point compactification of  $\R^3$.
The reason of this change will be made clear in Remark \ref{rem:whyS3}. For this section, we stick to the standard definition.

\subsection{Fox, Dehn and Alexander-Briggs Colorings}
\label{subsec:foxdehns}

The number of crossings of a knot diagram determines the number of strands and regions. 

\begin{lemma} (see \cite{alexander1928topological})
\label{lem:knot_diag}
Let $D$ be a knot diagram with $n$ crossings. Then it has $n$ strands and~$n+2$ regions.
\end{lemma}

The colorings of the strands of the knot diagram are called Fox colorings, and the colorings of the regions are called Dehn colorings. 

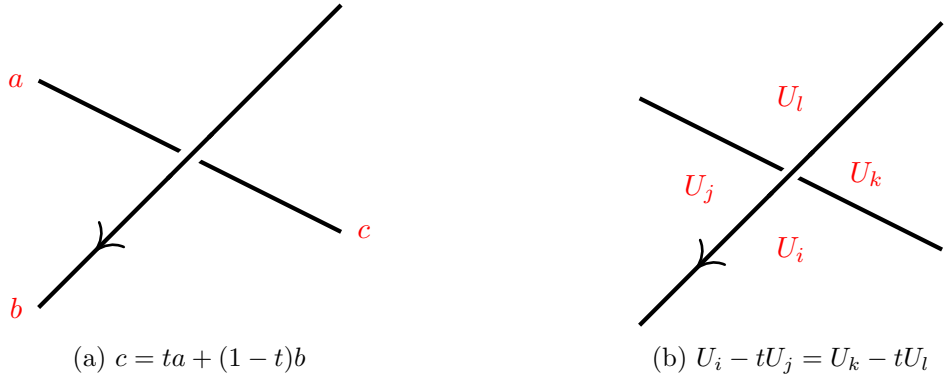
\begin{figure}[h!]
    \centering
\begin{subfigure}{.4\textwidth}
\centering

\begin{tikzpicture}[pics/arrow/.style={code={%
  \draw[line width=0pt,{Computer Modern Rightarrow[line
  width=1pt,width=3ex,length=2ex]}-] (-0.5ex,0) -- (0.5ex,0);
  }}]
\begin{knot}[
clip width=5,
]
\strand[black,ultra thick] (-2,-2) to pic[pos=0.2,sloped]{arrow} (2,2);
\strand[black,ultra thick] (-2,1) -- (2,-1);;
\end{knot}

\tikzset{nnode/.style = {shape=circle,fill=myg,draw,inner sep=0pt,minimum
size=1.9em}}
\node[text=red] at (-2.3,1) {$a$};
\node[text=red] at (2.3,-1) {$c$};
\node[text=red] at (-2.3,-2) {$b$};

\end{tikzpicture}
    \caption{$c=ta +(1-t)b$}
    \label{fig:Fox_coloring}
\end{subfigure}
\hspace{0.1\textwidth}
\begin{subfigure}{.4\textwidth}
\centering
\begin{tikzpicture}[pics/arrow/.style={code={%
  \draw[line width=0pt,{Computer Modern Rightarrow[line
  width=1pt,width=3ex,length=2ex]}-] (-0.5ex,0) -- (0.5ex,0);
  }}]
\tikzset{nnode/.style = {shape=circle,fill=myg,draw,inner sep=0pt,minimum
size=1.9em}}

\begin{knot}[
clip width=4,
]
\strand[black,ultra thick] (-2,-2) to pic[pos=0.2,sloped]{arrow} (2,2);
\strand[black,ultra thick] (-2,1) -- (2,-1);;
\end{knot}

\node[text=red] at (0,-1) {$U_i$};
\node[text=red] at (0,1) {$U_l$};
\node[text=red, ultra thick] at (-1.2,-0.2) {$U_j$};
\node[text=red] at (1,0) {$U_k$};

\end{tikzpicture}
    \caption{$U_i-tU_j=U_k-tU_l$}
    \label{fig:Dehn_coloring}
\end{subfigure}

\caption{Fox coloring \ref{fig:Fox_coloring} and Dehn coloring \ref{fig:Dehn_coloring} of knot diagrams. }
\label{fig:Fox-Dehn}
\end{figure}

\begin{definition}
\label{def:FoxRt}
A \textbf{Fox $(R,t)$-coloring} of a knot diagram is a coloring of its strands such that for each crossing 
\begin{equation}
\label{eq:Fox}
c=ta +(1-t)b,
\end{equation}
where the strand with color $b$ is the overstrand and the strands colored with $a$ and $c$ are understands such that the rotation from the strand that is colored with $b$ to the strand that is colored with $c$ is counter clockwise; see Figure \ref{fig:Fox_coloring}.
A coloring is called trivial if all the colors are the same. The knot diagram is called \textbf{Fox $(R,t)$-colorable} if there exists a non-trivial coloring.
\end{definition}

The equations of the form \eqref{eq:Fox} can be stored in a matrix, whose kernel is called \textbf{the \textbf{module of Fox $(R,t)$-colorings}}.

\begin{definition}
\label{def:coloringMat}
Consider a knot diagram $D$ with crossings labeled as $\{c_1,\ldots,c_n\}$, and strands labeled as~$\{x_1,\ldots,x_n\}$. The \textbf{Fox coloring matrix} of $D$ is the matrix $M(t)$ is defined as follows.
$$
M_{ij}(t) = 
\begin{cases}
1 - t & \text{if} \ \ x_j \textnormal{ is an overstrand at } c_i,\\
-1 & \text{if} \ \ x_j \textnormal{ is an understrand at } c_i \textnormal{ at the left side of the overstrand},\\
t & \text{if} \ \ x_j \textnormal{ is an understrand at } c_i \textnormal{ at the right side of the overstrand},\\
0 & \text{otherwise}
\end{cases}
$$ for $1 \le i,j \le n$.
\end{definition}

We follow a similar procedure for the coloring of the regions a knot diagram.

\begin{definition}
\label{def:DehnRt}
A \textbf{Dehn $(R,t)$-coloring} of a knot diagram with $n$ crossing is a coloring of its regions such that for each crossing
\begin{equation}
\label{eq:Dehn}
U_i - tU_j = U_k - tU_l,
\end{equation}
where the regions $U_i,U_j,U_k$ and $U_l$ are regions that have the same crossing on their border in a way that $U_i$ and $U_k$ are on the left side of the overstrand and~$U_j$ and~$U_l$ are on the right side of the overstrand; see Figure~\ref{fig:Dehn_coloring}. We assume that the color of the outside region is the color $0$. The knot diagram is called \textbf{Dehn $(R,t)$-colorable} if there exists a non-trivial coloring.
\end{definition}

The equations of the form \eqref{eq:Dehn} can be stored in a matrix, whose kernel is called \textbf{the \textbf{module of Dehn $(R,t)$-colorings}}.

\begin{definition}
\label{def:DehnMat}
Consider a knot diagram $D$ with crossings labeled as $\{c_1,\ldots,c_n\}$, and regions labeled as~$\{U_1,\ldots,U_{n+2}\}$. The \textbf{Dehn coloring matrix}~$N(t)$ of $D$ is defined as follows.

$$N_{ms}(t) = 
\begin{cases}
1 & \text{if} \ \ s=i,\\
-t & \text{if} \ \ s=j,\\
-1 & \text{if} \ \ s=k,\\
t & \text{if} \ \  s=l,\\
0 & \text{otherwise},
\end{cases}
$$ for $1 \le m \le n$ and $1 \le s \le n+2$ with $U_i$, $U_j$, $U_k$ and $U_l$ are at crossing $c_m$ where the entries of $N(t)$ are determined according to the rule in Definition \ref{def:DehnRt}.

\end{definition}

Before moving on to the next coloring, we define the elementary ideals of a given a matrix over the ring $R$, and note that one of the most famous knot invariants, the Alexander polynomial, can be defined in terms of an elementary ideal of the Fox coloring matrix.

\begin{definition}\label{d-elem-id}
Let $K$ be a knot, $A \in R^{m \times n}$ and $k \in \Z_{\geq 0}$. The ideal generated by the determinants of all $(n-k)\times (n-k)$ submatrices of $A$ is called the $k$-th \textbf{elementary ideal} of $A$, and is denoted by $E_k(A)$. We say $E_k(A)=0$ if $n-k>m$, and $E_k(A)=R$ if $n-k\leq 0$. The \textbf{Alexander polynomial} of $K$, denoted by $\Delta_K(T)$, 
is the generator of $E_1(M(T))$ such that the constant term is positive, where $M(T)$ is the Fox coloring matrix (see Definition \ref{def:coloringMat}).
\end{definition}

Lastly, the coloring of the vertices of the Tait diagram is called Alexander-Briggs (AB) colorings.

\begin{definition}
\label{def:AB-coloring}
Given a Tait diagram of a knot, let
$$
\wt (U) = \sum_{v \in \partial U} \wt (U,v) v
$$
where $\partial U$ is the boundary of the region $U$ of the Tait diagram, and $v$ is a vertex lying on the boundary, and 
$$\wt (U,v) = 
\begin{cases}
t & \text{if there is a dot in } U \textnormal{ next to } v,\\
1 & \text{otherwise}. 
\end{cases}
$$
An \textbf{Alexander-Briggs $(R,t)$-coloring} of a Tait diagram is a coloring of its vertices such that $\wt(U) = 0$ for all regions $U$ of the Tait diagram. The Tait diagram is called \textbf{Alexander-Briggs $(R,t)$-colorable} if there exists a nonzero Alexander-Briggs $(R,t)$-coloring.  
\end{definition}

Similar to Fox and Dehn coloring matrices defined in this subsection, one can define the AB coloring matrix. 

\begin{definition}
\label{def:ABmatrix}
Consider a Tait diagram with vertices labeled as $\{v_1,\ldots,v_n\}$, and regions labeled as $\{U_1,\ldots,U_{n+2}\}$. The \textbf{Alexander-Briggs $(R,t)$-coloring matrix} of the Tait diagram is the matrix $P(t)$ defined as follows.

$$P_{rs}(t) = \wt(U_r,v_s)$$ for $1 \le r \le n+2$ and $1 \le s \le n$ where $v_s$ is in the boundary of $U_r$.
\end{definition}

In \cite{kilic2024knot}, it is noted that the three notions of colorability coincide: a Tait diagram is Fox $(R,t)$-colorable if and only if it is Dehn $(R,t)$-colorable if and only if it is Alexander-Briggs $(R,t)$-colorable.

\begin{proposition}({see \cite{kilic2024knot}})
\label{r-trivial-Dehn-col}
A knot diagram is Dehn $(R,t)$-colorable if and only it is Fox $(R,t)$-colorable if and only it is Alexander-Briggs $(R,t)$-colorable.
\end{proposition}

\subsection{Operations on Graphs and Codes}

In this subsection we recall some operations on graphs and codes that will be useful in the later parts of this paper.

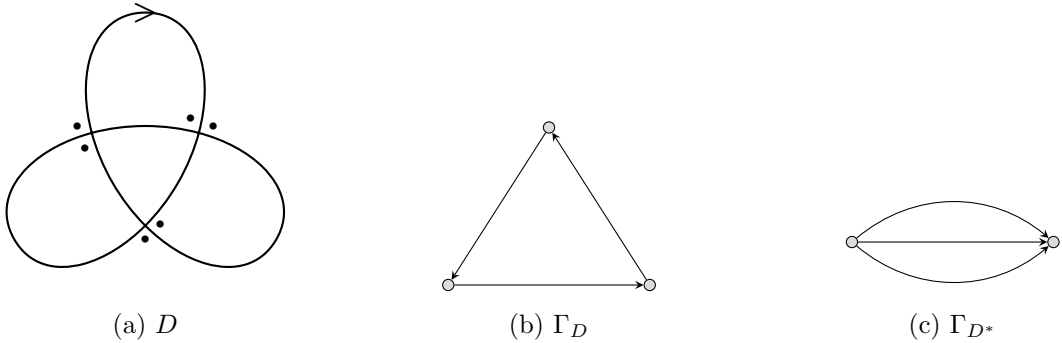
\begin{figure}[ht!]
    \centering
\subcaptionbox{$D$}
[.3\linewidth]{\begin{tikzpicture}[use Hobby shortcut,
every trefoil component/.style={thick, draw}, pics/arrow/.style={code={%
  \draw[line width=0pt,{Computer Modern Rightarrow[line
  width=1pt,width=3ex,length=2ex]}-] (-0.5ex,0) -- (0.5ex,0);
  }}]
\path[spath/save=trefoil] ([closed]90:2) foreach \k in {1,...,3} {
.. (-30+\k*240:.5) .. (90+\k*240:2) } (90:2);
\tikzset{spath/knot={trefoil}{0pt}{1,3,5}}
\node[text=black] at (0,2) {$\boldsymbol{>}$};

\node[text=black] at (0.9,0.5) {\tiny{$\bullet$}};
\node[text=black] at (0.6,0.6) {\tiny{$\bullet$}};

\node[text=black] at (0.2,-0.8) {\tiny{$\bullet$}};
\node[text=black] at (0,-1) {\tiny{$\bullet$}};

\node[text=black] at (-0.9,0.5) {\tiny{$\bullet$}};
\node[text=black] at (-0.8,0.2) {\tiny{$\bullet$}};
\end{tikzpicture}
}
 \hspace{.03\textwidth}
\subcaptionbox{$\Gamma _D$}
[.3\linewidth]{
\begin{tikzpicture}
\tikzset{nnode/.style = {shape=circle,fill=myg,draw,inner sep=1.5pt, minimum
size=0.2em}}
\tikzset{edge/.style = {->,> = stealth}}

\node[nnode] (S1) {};
\node[shape=coordinate,right=0.5\mynodespace of S1] (K) {};
\node[nnode,right=\mynodespace of S1] (S2) {};
\node[nnode,above=0.8\mynodespace of K] (S3) {};

\draw[edge,bend left=0] (S3)  to node{} (S1);

\draw[edge,bend left=0] (S1)  to node{} (S2);

\draw[edge,bend right=0] (S2)  to node{} (S3);

\end{tikzpicture}}
 \hspace{.03\textwidth}
\subcaptionbox{$\Gamma _{D^*}$}
[.3\linewidth]{\begin{tikzpicture}
\tikzset{nnode/.style = {shape=circle,fill=myg,draw,inner sep=1.5pt, minimum
size=0.2em}}
\tikzset{edge/.style = {->,> = stealth}}

\node[nnode] (S1) {};
\node[nnode,right=\mynodespace of S1] (S2) {};

\draw[edge,bend left=0] (S1)  to node{} (S2);

\draw[edge,bend left=40] (S1)  to node{} (S2);

\draw[edge,bend right=40] (S1)  to node{} (S2);

\end{tikzpicture}}
\caption{A Tait diagram $D$ of the Trefoil knot, and its black and white graphs where the outside region is colored white.  \label{fig:signedgraph}}
\end{figure}

\begin{definition} {\cite[Definition 6.1]{kilic2024knot}}
\label{def:signedgraph}
Let $D$ be a Tait diagram of a knot, and $D^*$ be the same diagram with the interchanged checkerboard coloring. The \textbf{black graph} of $D$ is the planar graph $\Gamma _D$ whose vertices are the black regions of $D$. For each crossing there is an edge between two vertices if the black regions in the Tait diagram corresponding to these vertices have that crossing in their common boundaries. Similarly, $\Gamma _{D^*}$ is called the \textbf{white graph} of $D$. The graphs can be made directed by choosing the direction from the region without a dot to the region with a dot near the crossing in their common boundary. See Figure \ref{fig:signedgraph} for illustration.
\end{definition}

One can associate a positive sign to a crossing, and therefore to the corresponding edge in the graph, 
if the rotation from the outgoing overstrand to the outgoing understrand is counter clockwise, and negative otherwise. In this way, we get a \textbf{signed graph}. The signs do not change if we take the opposite orientation of the knot. Hence, from any Tait diagram~$D$ of an oriented knot, once the black graph is constructed, the associated directed graph, signed graph, and directed signed graph are obtained automatically and we denote all by~$\Gamma _D$. Similarly for the white graph $\Gamma _{D^*}$.

\begin{remark}
\label{r-altern-black-white}
By \cite[Proposition 3.2]{kauffman1989}, the collection of connected planar link diagrams is in one-to-one 
correspondence with the collection of connected signed planar graphs. Moreover, 
the signed graph has all signs of the same type, if and only if the link 
diagram is alternating.

\end{remark}

One can define graph codes from black and white directed graphs of the Tait diagram of knots, or simply from any directed graph, by using their incidence matrices.

\begin{definition}
\label{def:graphcode}
Consider a directed graph $\Gamma$ with vertices $\{v_1,\ldots,v_m\}$ and edges $\{e_1,\ldots ,e_n\}$. The matrix $A(t)$ of $\Gamma$ is defined as follows.
$$
A(t)_{ij} = 
\begin{cases}
1 & \text{if} \ \ e_j \textnormal{ is an outgoing edge of } v_i,\\
t & \text{if} \ \ e_j \textnormal{ is an ingoing edge of } v_i,\\
0 & \text{otherwise}.
\end{cases}
$$ for $1 \le i \le m$ and $1 \le j \le n$. The code with the parity check matrix $A(t)$ is denoted by~$\mC_{\Gamma,t}$.
\end{definition}

\begin{remark}\label{r-A(t)}
Observe that when $t=-1$ in Definition \ref{def:graphcode}, the matrix $A(-1)$ is equal to the incidence matrix of $\Gamma$. Moreover, deleting one of its rows still induces a parity check matrix for the code $\mC_{\Gamma,-1}$ by the property that if $t=-1$, then the rows are linearly dependent as the sum of them gives the zero vector. We will use $\mC_{\Gamma}$ instead of $\mC_{\Gamma,-1}$ for the ease of notation. The code $\mC_{\Gamma}$ is of particular interest, and is called the \textbf{cycle code} of $\Gamma$. 
\end{remark}

We continue with some well-known codes and graphs that are obtained via basic operations. We refer to \cite{pellikaan2018} for more details.

\begin{definition}\label{d-contract-delete}
Consider an edge $e$ of the graph $\Gamma$.
The \textbf{deletion graph} $\Gamma \setminus e$ of $\Gamma$ at $e$ is the graph obtained from $\Gamma$ by deleting $e$. The \textbf{contraction graph} $\Gamma / e$ of $\Gamma $ at $e$ is the graph obtained from $\Gamma$ by deleting the edge $e$ and contracting the endpoints of $e$ to a single vertex.
\end{definition}

Philosophically similar to operations mentioned in Definition \ref{d-contract-delete} for graphs, there are the puncturing and shortening operations for codes.

\begin{definition}\label{d-punct-short}
Consider a coordinate $i$ of the code $\mC $. The \textbf{punctured code} $\mC_i$ of $\mC$ at the $i$-th coordinate is the code obtained from $\mC$ by deleting the $i$-th coordinate of the codewords. The \textbf{shortened code} $\mC^i$ of $\mC$ at the $i$-th coordinate is the code obtained from~$\mC$ by puncturing the $i$-th coordinate of all codewords that have a zero at the $i$-th coordinate.

\end{definition}

The following result would play a role in proving the main result of this paper.

\begin{proposition}\label{p-code-punct-short}
Let $\mC$ and $\mD$ be $\F_q$-linear codes of the same length. 
Suppose that the shortened codes of $\mC$ and $\mD$  at the $i$-th position are the same, and also the punctured codes of $\mC$ and $\mD$  at the $i$-th position are the same.
Then $\mC=\mD$.
\end{proposition}

\begin{proof}
Without loss of generality, we may assume that $i=1$.
Let $\mC^1$ and $\mD^1$ be the shortened codes of $\mC$ and $\mD$  at the first position.
Let $\mC_1$ and $\mD_1$ be the punctured codes of $\mC$ and $\mD$  at the first position.
Then $\mC^1=\mD^1$  and  $\mC_1=\mD_1$ by assumption.

If $\mC$ is degenerate at the first position, that is, all codewords of $\mC$ are zero at the first position, then $\mC^1=\mC_1$.
So also $\mD^1=\mD_1$ and all codewords of $\mD$ are zero at the first position.
Hence $\mC= \{ (0,x) | x \in \mC^1 \}$ and $\mD= \{ (0,y) | y \in \mD^1 \}$.
Therefore $\mC=\mD$, since~$\mC^1=\mD^1$.

If $\mC$ is non-degenerate at the first position, then there is a codeword of $\mC$ that is not zero at the first position and $\mC^1\not=\mC_1$.
Let $ c\in \mC_1\setminus \mC^1 $. Then also $\mD^1\not=\mD_1$ and $ c\in \mD_1\setminus \mD^1 $, since $\mC^1=\mD^1$  and  $\mC_1=\mD_1$.
So $\mC$ is generated by $(1,c)$ and $\{ (0,x) | x \in \mC^1 \}$. Similarly, $\mD$ is generated by $(1,c)$ and $\{ (0,y) | y \in \mD^1 \}$.
Therefore $\mC=\mD$, since $\mC^1=\mD^1$.
\end{proof}

Puncturing and shortening behave in a complementary way with respect to duality as the following relation shows.

\begin{proposition} ({see \cite[Proposition 2.1.27]{pellikaan2018}})
\label{p-contr-}
Let $i$ be a coordinate of the code $\mC $. We have $(\mC_i)^\perp=(\mC^\perp)^i$ and $(\mC^i)^\perp=(\mC^\perp)_i$.
\end{proposition}

The graph version of Proposition \ref{p-contr-} is the following.

\begin{proposition}
\label{p-contr-del-punct-short}
Let $e$ be an edge of the graph $\Gamma$, and $i$ be the coordinate of $\mC_{\Gamma}$ that corresponds to the edge $e$. We have $\mC_{\Gamma \setminus e}= (\mC_{\Gamma})^i$ and $\mC_{\Gamma / e}= (\mC_{\Gamma})_i$.
\end{proposition}

We conclude this subsection by recalling the codes arising from the black and white graphs of a Tait diagram, together with a crucial result establishing their duality under certain assumptions.

\begin{definition}
\label{def:blackcode}
The codes $\mC_{\Gamma_D,t}$ and $\mC_{\Gamma_{D^*},t}$ of the black and white graphs defined in Definition \ref{def:signedgraph} are called the \textbf{black code} and the \textbf{white code}, denoted by $\mC_{D,t}$ and $\mC_{D^*,t}$, respectively. When $t=-1$, we simply denote them as $\mC_{D}$ and $\mC_{D^*}$.

\end{definition}

\begin{theorem}(see \cite[Theorem 6.6]{kilic2024knot})
 \label{p-dual-black-white}
 Let $D$ be a reduced Tait diagram of a knot.
 If the characteristic is $2$ or the diagram is alternating, 
 then the black and white codes are dual to each other, that is, $\mC_D^\perp =\mC_{D^*}$.
 \end{theorem}

The path to go from knot colorings to codes was to use their coloring matrices. More precisely, in \cite{kilic2024knot}, the properties of so-called knot codes were studied and these codes were either Fox codes, Dehn codes or AB codes. Each of these codes was constructed by taking the corresponding Fox, Dehn or AB coloring matrix as the parity-check matrix of the resulting codes.

We conclude this subsection by recalling the following nice relation between AB codes and the intersection of the black and white codes of a Tait diagram $D$. Before stating the result we note that the \textbf{hull} of a code is its intersection with its dual.

\begin{proposition}(see \cite[Proposition 6.7]{kilic2024knot})
\label{p-AB-hull}
 The code $\mC_{D,t} \cap \mC_{D^*,t}$ is equal to the Alexander-Briggs code $\mA_{D,t}$. 
 If $t=-1$, then $\mA_{D}$ is equal to the hull of $\mC_{D}$.
 \end{proposition}

\section{Tait's Flyping Conjecture}
\label{sec:Tait}

In his foundational papers \cite{tait1898}, Tait discovered certain type of operation on alternating link diagrams, called \textit{a flype}. We refer to \cite{menasco1993} for the history and the usage of this term within knot theory.

For the rest of this paper, we will assume that a knot is an embedding of the unit circle $S^1$ in $S^3$ (not in $\R^3$ as in Definition \ref{def:knot} we used earlier). We will explain the reason for this change in Remark \ref{rem:whyS3}. Moreover, for the rest of this section, a \textbf{circle} and \textbf{disk} in $\R^2$ or the $2$-sphere, and a \textbf{sphere} and \textbf{ball} in $\R^3$ or the $3$-sphere, will mean the image of a tame embedding of the standard circle 
$\{(x_1,x_2) \mid x_1, x_2 \in \R, \, x_1^2+x_2^2=1\}$, disk $\{(x_1,x_2) \mid x_1, x_2 \in \R, \, x_1^2+x_2^2\leq 1\}$, 
sphere $\{(x_1,x_2,x_3) \mid x_1, x_2, x_3 \in \R, \, x_1^2+x_2^2+x_3^2=1\}$, 
and ball $\{(x_1,x_2,x_3) \mid x_1, x_2, x_3 \in \R, \, x_1^2+x_2^2+x_3^2 \leq 1\}$, respectively.


\begin{definition}
Let $S^2 = \mathbb{R}^2 \cup \{\infty\}$ and $S^3 = \mathbb{R}^3 \cup \{\infty\}$. Let $C\subset S^2$ be a circle such that~$\infty \notin C$. Then, $S^2 \setminus C$ has two connected components. The component
that contains~$\infty$ is called the \textbf{outside} of $C$, and the other component is
called the \textbf{inside} of $C$. Similarly, let $S\subset S^3$ be a sphere such that $\infty\notin S$. Then $S^3\setminus S$ has two connected components. The
component that contains $\infty$ is called the \textbf{outside} of $S$, and the other
component is called the \textbf{inside} of $S$.  
\end{definition}

We continue by defining a portion of a knot diagram, first introduced by Conway in \cite{conway1970enumeration} as a classical knot-diagram notion, which is a special case of a more general topological notion.

\begin{definition}
\label{def:ntangle}
An \textbf{$n$-tangle} is a proper embedding of the disjoint union of n arcs into a ball such that the embedding must send the endpoints of the arcs to $2n$ specified marked points on the ball's boundary. A diagram of an $n$-tangle is a regular projection of that $n$-tangle on a disk such that the $2n$ marked endpoints on the boundary of the ball are projected to $2n$ marked points on the boundary of the disk. If the arcs are oriented, we have $n$ incoming and 
$n$ outgoing marked boundary points. In particular, a $2$-tangle is simply called a \textbf{tangle}, following Conway's notation.
\end{definition}

Connecting the two endpoints of a diagram of a $1$-tangle by an embedded arc on the outside of the disk of the diagram of the $1$-tangle gives a 
diagram of a knot, and we call the diagram the \textbf{closure} of the tangle. Conversely, if there exists a circle $C$ on the plane that intersects the diagram $D$ of a knot in exactly two points, then it divides $D$ into two $1$-tangles, one within the circle $C$ and the other with a $1$-tangle on the outside.

\begin{definition} 
Let $D_i$ be a diagram of an oriented knot $K_i$, and $T_i$ be a $1$-tangle in a disk with boundary circle $C_i$ such that $D_i$ is the closure of $T_i$ and its closing arc is a subset of the arc $x_i$ between two consecutive crossings of $D_i$, for $i=1,2$. Assume furthermore that~$C_1$ and $C_2$ coincide,
and that the inside disk of $C_1$ is equal to the outside disk of $C_2$, and the inside disk of $C_2$ is equal to the outside disk of $C_1$;
and the marked points coincide such that the outgoing endpoint of $C_1$ is the incoming endpoint of $C_2$, and the outgoing endpoint of~$C_2$ is the incoming endpoint of $C_1$. Then the union of $T_1$ and $T_2$ gives a diagram of a knot, that is the oriented \textbf{connected sum} $K_1 \# K_2$ of the oriented knots $K_1$ and $K_2$,
and will be denoted by $(D_1,x_1) \# (D_2,x_2)$.\\
A knot that cannot be written as the sum of two non-trivial knots is called a \textbf{prime} knot, otherwise it is called a \textbf{composite} knot.
\end{definition}

The connected sum of knots does not depend on the choice of the arcs \cite{rolfsen2003knots},
and Tait's flyping conjecture shows that the equivalent diagrams could be obtained by a series of flypes in the case of prime alternating knots.

\textbf{Flyping} is an operation on a diagram $D$ of an alternating knot that consists of $3$ tangles denoted by $T$, $E$ and $R$. An alternating knot has the property that for the black regions the direction from the overstrand to the understrand is counterclockwise for all crossings, see e.g. \cite{lickorish1997}.

Tangle $T$ has black regions $a$ and $b$ on the left-hand and right-hand side, respectively. The tangle $E$ consists of a simple twist with black regions $b$ and $c$ on the left-hand and right-hand side, respectively.
Tangle $R$ has black regions $c$ and $d$ on the left-hand and right-hand side, respectively,
where $d$ is the same region as $a$ of the tangle $T$.
The corresponding graph $\Gamma=\Gamma_D$ is also given in Figure \ref{fig:flype1}.

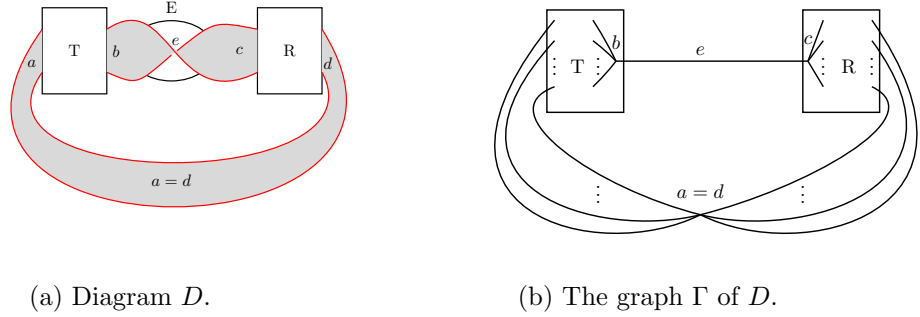
\begin{figure}[htbp]
\centering
\begin{subfigure}{0.44\textwidth}
\centering
\resizebox{1.2\textwidth}{!}{
\begin{tikzpicture}
    \draw[thick] (0,1) rectangle (1.5,-1);
    \node at (0.75,0) {T};

    \draw[thick] (5,1) rectangle (6.5,-1);
    \node at (5.75,0) {R};

    \draw[thick] (3,0) ellipse (1 and 0.7);
    \node at (3,1) {E};

    \begin{scope}
        \fill[gray!30] 
            (0,0.5) .. controls (-4,-5) and (10,-5) .. (6.5,0.5) 
            -- (6.5,-0.5) .. controls (8,-3.3) and (-2,-3.3) .. (0,-0.5) -- cycle;
    \end{scope}
    \begin{scope}
        \fill[gray!30]
            (1.5,0.5) .. controls (3,1.5) and (3,-1.5) .. (5,-0.5)
            -- (5,0.5) .. controls (4,0.7).. (3.13,0)
            -- (3,-0.1) .. controls (2.2,-0.8)  .. (1.5,-0.5) -- cycle;
    \end{scope}
   
    \draw[red, thick]
        (1.5,0.5) 
        .. controls (3,1.5) and (3,-1.5) .. (5,-0.5); 

    \draw[red, thick]
        (1.5,-0.5) 
        .. controls (2.2,-0.8) .. (3,-0.1); 
    \draw[red, thick]
        (3.13,0) 
        .. controls (4,0.7).. (5,0.5); 
    \draw[red, thick]
        (0,0.5) 
        .. controls (-4,-5) and (10,-5) .. (6.5,0.5); 
    \draw[red, thick]
        (0,-0.5) 
        .. controls (-2,-3.3) and (8,-3.3) .. (6.5,-0.5); 

    \node[left] at (0,-0.3) {$a$};
    \node[right] at (6.4,-0.3) {$d$};
    \node[above] at (3.1,0) {$e$};
    \node[right] at (1.5,0) {$b$};
    \node[left] at (4.8,0) {$c$};

    \node at (3,-3) {$a=d$};
\end{tikzpicture}
}
\caption{Diagram $D$.}
\label{fig:flype1_D}
\end{subfigure}
\begin{subfigure}{0.45\textwidth}
\centering
\resizebox{1.2\textwidth}{!}{
\begin{tikzpicture}
    \draw[thick] (0,1) rectangle (1.5,-1);
    \node at (0.6,-0.1) {T};

    \draw[thick] (5,1) rectangle (6.5,-1);
    \node at (5.9,-0.1) {R};

    \draw[black, thick]
    (0.15,0.8)  
    .. controls (-3,-3) and (1,-4) .. (3,-3);
    \draw[black, thick]
    (0.15,0.4)  
    .. controls (-2.5,-2.5) and (1,-3.5) .. (3,-3);
    \draw[black, thick]
    (0.15,-0.5)  
    .. controls (-1.2,-1) and (1,-2.5) .. (3,-3);
    \draw[black, thick]
    (3,-3)  
    .. controls (5,-2.5) and (7.6,-1.1) .. (6.35,-0.5);
    \draw[black, thick]
    (3,-3)  
    .. controls (5,-3.5) and (8.5,-2.5) .. (6.35,0.4);
    \draw[black, thick]
    (3,-3)  
    .. controls (5,-4) and (9,-3) .. (6.35,0.8);

    \draw[black, thick]
    (0.9,0.8) .. controls (1.15,0.4) .. (1.35,0);
    \draw[black, thick]
    (0.9,0.4) .. controls (1.15,0.2) .. (1.35,0);
    \draw[black, thick]
    (0.9,-0.5) .. controls (1.15,-0.25) .. (1.35,0);

    \draw[black, thick]
    (5.4,0.8) .. controls (5.25,0.4) .. (5.1,0);
    \draw[black, thick]
    (5.4,0.4) .. controls (5.25,0.2) .. (5.1,0);
    \draw[black, thick]
    (5.4,-0.5) .. controls (5.25,-0.25) .. (5.1,0);

    \draw[black, thick]
    (5.1, 0) .. controls (1.88,0) .. (1.35,0);

    \node at (5.4,0) {$\vdots$};
    \node at (1,0) {$\vdots$};
    \node at (5,-2.5) {$\vdots$};
    \node at (1,-2.5) {$\vdots$};
    \node at (6.35,0) {$\vdots$};
    \node at (0.15,0) {$\vdots$};
    \node at (3,0.2) {$e$};

    \node[above] at (5.1,0.2) {$c$};
    \node[above] at (1.35,0.1) {$b$};
    \node[above] at (3,-2.8) {$a=d$};
\end{tikzpicture}
}
\caption{The graph $\Gamma$ of $D$.}
\label{fig:flype1_T}
\end{subfigure}
 \caption{Diagram $D$ and its graph $\Gamma$ before flype.}
 \label{fig:flype1}
\end{figure}

If we fix $R$ and turn the tangle $T$ along the horizontal axis $180$ degrees so that the twist on the right-hand side is undone, 
we get the tangle $\perp$ and the Tait diagram $D'$ of an equivalent alternating link, with a new twist on the left-hand side, which is again denoted by $E$. 
See Figure \ref{fig:flype2} with tangles $E$, $\perp$ and $R$ and the corresponding graph~$\Gamma '=\Gamma_{D'}$. The tangle $E$ consists of a simple twist with black regions $d$ and $a$ on the left-hand and right-hand side, respectively. 
Tangle $\perp$ has black regions $a$ and $b$ on the left-hand and right-hand side, respectively.
Tangle $R$ has black regions $c$ and $d$ on the left-hand and right-hand side, respectively,
where $c$ is the same region as $b$ of the tangle $\perp$.
The corresponding graph~$\Gamma'=\Gamma_{D'}$ is also given in Figure \ref{fig:flype2}.

\begin{figure}[htbp]
\centering
\begin{subfigure}{0.49\textwidth}
\centering
\resizebox{1.1\textwidth}{!}{
\begin{tikzpicture}
\draw[thick] (2.5,1) rectangle (4,-1);
    \node at (3.25,0) {\rotatebox{180}{T}};

    \draw[thick] (5.5,1) rectangle (7,-1);
    \node at (6.3,0) {R};

    \draw[thick] (0.75,0) ellipse (1 and 0.7);
    \node at (0.75,1) {E};

    \begin{scope}
        \fill[gray!30] 
            (7,0.5) .. controls (10,-5) and (-4,-5) .. (-0.5,0) 
            .. controls (0,0.5) and (1.5,-0.5) .. (2.5,-0.5) 
            -- (2.5,0.5) .. controls (1.5,0.5) ..(0.5,0.1) 
            -- (0.4,0)
            .. controls (0.2,-0.2) .. (0,-0.5)
            .. controls (-2,-3.3) and (8,-3.3) .. (7,-0.5) -- cycle;
    \end{scope}

    \fill[gray!30] 
        (3.8,-0.5) -- (5.7,-0.5) 
        -- (5.7,0.5) -- (3.8,0.5) -- cycle;

    \draw[red, thick] (5.7,-0.5) -- (3.8,-0.5);
    \draw[red, line width=0.4mm] (5.7,0.5) -- (3.8,0.5);

    \draw[red, thick]
    (7,0.5) .. controls (10,-5) and (-4,-5) .. (-0.5,0) 
    .. controls (0,0.5) and (1.5,-0.5) .. (2.5,-0.5);
    \draw[red, thick]
    (7,-0.5) .. controls (8,-3.3) and (-2,-3.3) .. (0,-0.5) 
    .. controls (0.2,-0.2).. (0.4,0);
    \draw[red, thick]
    (0.5,0.1) .. controls (1.5,0.5) .. (2.5,0.5);
    \node[right] at (4.2,0) {$b=c$};

    \node at (-0.4,-0.5) {$d$};
    \node at (2.1,0) {$a$};
     \node at (7.2,-0.5) {$d$};
\end{tikzpicture}
}
\caption{Diagram $D'$.}
\label{fig:flype2_D}
\end{subfigure}
\hfill 
\begin{subfigure}{0.49\textwidth}
\centering
\resizebox{1.1\textwidth}{!}{
\begin{tikzpicture}
\draw[thick] (0,1) rectangle (1.5,-1);
    \node at (0.6,-0.1) {\rotatebox{180}{T}};

    \draw[thick] (5,1) rectangle (6.5,-1);
    \node at (5.9,-0.1) {R};

    \draw[black, thick] (0.15,0.8) -- (-1,0);
    \draw[black, thick] (0.15,0.4) -- (-1,0);
    \draw[black, thick] (0.15,-0.5) -- (-1,0);
    \draw[black, thick] (6.35,0.8) -- (7.2,0);
    \draw[black, thick] (6.35,0.4) -- (7.2,0);
    \draw[black, thick] (6.35,-0.5) -- (7.2,0);

    \draw[black, thick]
    (-1,0)  
    .. controls (-3,-3) and (9,-3) .. (7.2,0);

    \draw[black, thick]
    (0.9,0.8) .. controls (1.2,0.6) .. (3,0);
    \draw[black, thick]
    (0.9,0.4) .. controls (1.2,0.2) .. (3,0);
    \draw[black, thick]
    (0.9,-0.5) .. controls (1.2,-0.3) .. (3,0);


    \draw[black, thick]
    (5.4,0.8) .. controls (5,0.6) .. (3,0);
    \draw[black, thick]
    (5.4,0.4) .. controls (5,0.2) .. (3,0);
    \draw[black, thick]
    (5.4,-0.5) .. controls (5,-0.3) .. (3,0);

    \node[above] at (7.2,0) {$d$};
    \node[above] at (-1,0) {$a$};
    
    \node at (5.4,0) {$\vdots$};
    \node at (1,0) {$\vdots$};
    
    \node at (6.35,0) {$\vdots$};
    \node at (0.15,0) {$\vdots$};
    \node at (3,0.5) {$b=c$};

    \node[above] at (3,-2.8) {$e$};
\end{tikzpicture}
}
\caption{The graph $\Gamma'$ of $D'$.}
\label{fig:flype2_T}
\end{subfigure}
 \caption{Diagram $D'$ and its graph $\Gamma'$ after flype.}
 \label{fig:flype2}
\end{figure}
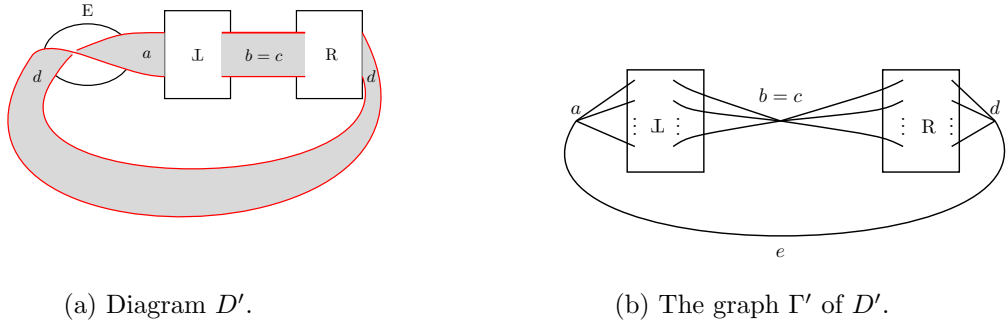


The nodes $a$ and $b$ are part of the sub-graph $T$. Similarly, this is the case for the sub-graph~$\perp$.  
The nodes $c$ and $d$ are part of the sub-graph $R$. 

In $D$ and the graph $\Gamma$ of $D$, $a$ and $d$ are identified, and in $D'$ and the graph $\Gamma'$ of $D'$, b and c are identified.

\begin{remark}
\label{r-flype-tangle}
The graphs $T$ and $\perp$ are isomorphic under the rotation, and they are equal when we identify the corresponding vertices and edges.
But the direction of the edges may be inverted. In fact, if the orientation of the knot is in/out of the black regions at the crossing, 
then the direction of the corresponding edge in the graph $T$ is inverted for the edge of $\perp $. 
If the orientation of the knot is in/in and out/out at the black regions at the crossing, then the direction of the edges of $T$ and
$\perp $ is the same.
\end{remark}

Similar to Definition \ref{def:Reide_equiv}, one can define equivalence with respect to flyping.

\begin{definition}
\label{def:Flype_equiv}
Two diagrams $D$ and $D'$ are called \textbf{flype equivalent} if $D$ can be transformed into~$D'$ by using a finite sequence of flype moves.
\end{definition}

The following theorem was conjectured by Tait, see \cite{tait1898}, \cite[\S 2]{thistlethwaite1985} and \cite[Chapter 11 \S 5]{murasugi1996knot}. It was proved by Menasco and Thistlethwaite in \cite{menasco1991-1} and \cite{menasco1993}. 

\begin{theorem}[Tait's Flyping Conjecture]
\label{thm-tait-conj}
Suppose $D_1$ and $D_2$ are two reduced alternating diagrams of a prime links $K_1$ and $K_2$ that are equivalent. Then $D_1$ and $D_2$ are flype equivalent.
\end{theorem}

Flype-equivalent diagrams are also Reidemeister equivalent (in the sense of Definition \ref{def:Reide_equiv}),
but the converse is not always true. What Tait's conjecture says is that Reidemeister equivalence is the same as flype equivalence for the class of
the reduced alternating diagrams of prime links. In the following remark, we now explain why the knot is considered in $S^3$ rather than $\R^3$.

\begin{remark}
\label{rem:whyS3}
In \cite{menasco1993}, where the Tait conjecture is proved, 
the knot is considered to lie in the three-sphere. 
If we were to consider the knots in $R^3$ instead of the 3-sphere, then the conjecture would have been false. Consider the diagram $D$ of the trefoil knot as in Figure \ref{fig:signedgraph},
and the diagram $D'$ as the same diagram $D$ but with the point at infinity in a bounded white region of $D$ next to the unbounded region. Then the black diagrams of $D$ and $D'$ are (Reidemeister) equivalent, but not flype equivalent.
\end{remark}

We conclude this section with an application of Theorem \ref{thm-tait-conj}.

\begin{proposition}
\label{p-tutte}
The 2-variable Tutte polynomial of the graph $\Gamma _D$ of a reduced Tait diagram $D$ of a prime alternating knot is an invariant. 
 \end{proposition}

\begin{proof}
It was noted by Thistlethwaite in \cite[page 316]{thistlethwaite1988} that the Tutte polynomial of the graph $\Gamma_D$ of a reduced Tait diagram of a prime alternating knot is invariant under flyping by the short remark:
``..., it is certainly true, and easy to check from its recurrence formula, that Tutte polynomial 
is invariant under Tait’s flyping operation.''. This was worked out in \cite[Proposition 2.3]{petersen2007}. 
\end{proof}

\section{Coding Theory Invariant for Prime Alternating Knots}
\label{sec:newinvariant}

The goal of this section is to construct a new invariant by using the Alexander–Briggs colorings and the associated codes. Building on the coding‑theoretic framework and Tait's flyping conjecture, we obtain the new invariant. This is done by showing that the ($\pm1$)-permutation equivalent class of the cycle code of a prime alternating knot is an invariant.

The reason for using the ($\pm1$)-permutation equivalence, rather than the more standard equivalences in coding theory, permutation or monomial, is explained next.

\begin{definition}
Two linear codes are called \textbf{($\pm1$)-permutation equivalent} if they can be obtained from each other by only permuting the coordinates and multiplying the coordinates with $\pm 1$. 
\end{definition}

\begin{remark}
\label{rem:hullofcodes}
In \cite[Remark 6.9]{kilic2024knot}, it is noted that if two codes are ($\pm1$)-permutation equivalent, then their hulls (intersection of the code with its dual) are also ($\pm1$)-permutation equivalent. This is the key property we need and it is not true for monomial equivalence of codes.     
\end{remark}

We start with the main theorem of this section, after a lemma needed in its proof. 

\begin{lemma} (see \cite[Proposition 7.13]{kilic2024knot})
\label{p-sum-code-graph-pol} 
Let $\Gamma_1$ be a directed graph with a node $p_1$ and $\Gamma_2$ be a directed graphs with a node $p_2$. Let $\Gamma =(\Gamma_1 \sqcup \Gamma_2 )/ (p_1,p_2)$ be a disjoint sum of these graphs obtained by identifying (gluing together) $p_1$ and $p_2$. We have
$$
C_{\Gamma}=C_{\Gamma_1} \oplus C_{\Gamma_2}
$$  
\end{lemma}

\begin{theorem}\label{thm-code-alternating}
The ($\pm1$)-permutation equivalence class of the cycle code $\mC_{\Gamma_D}$ of the graph~$\Gamma_D$ of a reduced Tait diagram $D$ of a prime alternating knot is 
an invariant of the knot.
\end{theorem}

\begin{proof}
By Theorem \ref{thm-tait-conj}, it is enough to show the invariance under a flype.
The proof follows the same pattern as the proof for the Tutte polynomial of alternating links, see Proposition \ref{p-tutte}.  Consider a flype from the Tait diagrams $D$ to $D'$ as explained and shown in the figures of Section \ref{sec:Tait}. 

The edge $e$ corresponds to the twist in the tangle $E$ between the regions $b$ and $c$ in diagram $D$, 
connecting the nodes $b$ and $c$ in the graph $\Gamma$
in Figure \ref{fig:flype1}. In diagram $D'$, the edge $e$ corresponds to the twist in the tangle $E$ between the regions $a$ and $d$, 
connecting the nodes $a$ and $d$ in the graph $\Gamma'$
in Figure \ref{fig:flype2}.

The matrix $A(-1)$ (see 
Definition \ref{def:graphcode}) is a parity check matrix of the code $\mC$, and the sum of the rows of $A(-1)$ is the all-zero vector. 
Deleting the row corresponding to $c$ from~$A(-1)$ gives a matrix that is still a parity check matrix of $\mC$ by Remark \ref{r-A(t)}.
Similarly, the matrix~$A'(-1)$ is a parity check matrix of the code $\mC'$,
and by deleting the row corresponding to $d$ from $A'(-1)$ gives a matrix that is still a parity check matrix of $\mC'$. 

Deleting the edge $e$ in $\Gamma $ and $\Gamma'$ gives the graphs $\Gamma \setminus e$ and $\Gamma' \setminus e$, respectively, 
as shown in Figure \ref{fig:flype3}.

\begin{figure}[htbp]
\centering
\begin{subfigure}{0.49\textwidth}
\centering
\resizebox{1\textwidth}{!}{
\begin{tikzpicture}
\draw[thick] (0,1) rectangle (1.5,-1);
    \node at (0.6,-0.1) {T};

    \draw[thick] (5,1) rectangle (6.5,-1);
    \node at (5.9,-0.1) {R};

    \draw[black, thick]
    (0.15,0.8)  
    .. controls (-3,-3) and (1,-4) .. (3,-3);
    \draw[black, thick]
    (0.15,0.4)  
    .. controls (-2.5,-2.5) and (1,-3.5) .. (3,-3);
    \draw[black, thick]
    (0.15,-0.5)  
    .. controls (-1.2,-1) and (1,-2.5) .. (3,-3);
    \draw[black, thick]
    (3,-3)  
    .. controls (5,-2.5) and (7.6,-1.1) .. (6.35,-0.5);
    \draw[black, thick]
    (3,-3)  
    .. controls (5,-3.5) and (8.5,-2.5) .. (6.35,0.4);
    \draw[black, thick]
    (3,-3)  
    .. controls (5,-4) and (9,-3) .. (6.35,0.8);
    
    \draw[black, thick]
    (0.9,0.8) .. controls (1.2,0.6) .. (2.2,0);
    \draw[black, thick]
    (0.9,0.4) .. controls (1.2,0.2) .. (2.2,0);
    \draw[black, thick]
    (0.9,-0.5) .. controls (1.2,-0.3) .. (2.2,0);

    \draw[black, thick]
    (5.4,0.8) .. controls (5,0.6) .. (4.2,0);
    \draw[black, thick]
    (5.4,0.4) .. controls (5,0.2) .. (4.2,0);
    \draw[black, thick]
    (5.4,-0.5) .. controls (5,-0.3) .. (4.2,0);

    \node at (5.4,0) {$\vdots$};
    \node at (1,0) {$\vdots$};
    \node at (5,-2.5) {$\vdots$};
    \node at (1,-2.5) {$\vdots$};
    \node at (6.35,0) {$\vdots$};
    \node at (0.15,0) {$\vdots$};

    \node[above] at (3,-2.8) {$a=d$};
     \node[above] at (2.2,0) {$b$};
      \node[above] at (4.2,0) {$c$};
\end{tikzpicture}
}
\caption{Graph $\Gamma \setminus e$.}
\label{fig:flype3_v1}
\end{subfigure}
\hfill 
\begin{subfigure}{0.46\textwidth}
\centering
\resizebox{1\textwidth}{!}{
\begin{tikzpicture}
\draw[thick] (0,1) rectangle (1.5,-1);
    \node at (0.6,-0.1) {\rotatebox{180}{T}};

    \draw[thick] (5,1) rectangle (6.5,-1);
    \node at (5.9,-0.1) {R};

    \draw[black, thick] (0.15,0.8) -- (-1,0);
    \draw[black, thick] (0.15,0.4) -- (-1,0);
    \draw[black, thick] (0.15,-0.5) -- (-1,0);
    \draw[black, thick] (6.35,0.8) -- (7.2,0);
    \draw[black, thick] (6.35,0.4) -- (7.2,0);
    \draw[black, thick] (6.35,-0.5) -- (7.2,0);

    \draw[black, thick]
    (0.9,0.8) .. controls (1.2,0.6) .. (3.2,0);
    \draw[black, thick]
    (0.9,0.4) .. controls (1.2,0.2) .. (3.2,0);
    \draw[black, thick]
    (0.9,-0.5) .. controls (1.2,-0.3) .. (3.2,0);

    \draw[black, thick]
    (5.4,0.8) .. controls (5,0.6) .. (3.2,0);
    \draw[black, thick]
    (5.4,0.4) .. controls (5,0.2) .. (3.2,0);
    \draw[black, thick]
    (5.4,-0.5) .. controls (5,-0.3) .. (3.2,0);

    \node[above] at (7.2,0) {$d$};
    \node[above] at (-1,0) {$a$};
    
    \node at (5.4,0) {$\vdots$};
    \node at (1,0) {$\vdots$};
    
    \node at (6.35,0) {$\vdots$};
    \node at (0.15,0) {$\vdots$};
    \node[above] at (3,0.2) {$b=c$};
\end{tikzpicture}
}
\caption{Graph $\Gamma' \setminus e$.}
\label{fig:flype3_v2}
\end{subfigure}
 \caption{Graphs $\Gamma \setminus e$ and $\Gamma' \setminus e$.}
 \label{fig:flype3}
\end{figure}
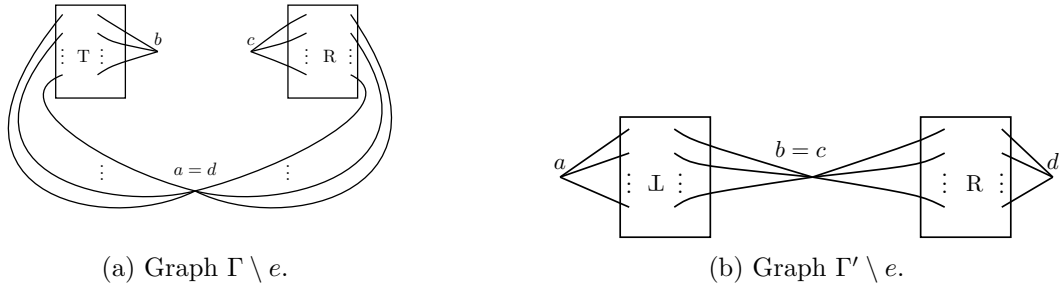

Then, $\Gamma \setminus e$ is the union of the subgraphs $T$ and $R$ that have the unique node $a=d$ in common.
So, $\Gamma \setminus e = (T\sqcup R)/(a,d)$.
Hence $ \mC_{\Gamma \setminus e}=\mC_T \oplus \mC_R $ by Lemma \ref{p-sum-code-graph-pol}, see Figure \ref{fig:flype3}.

Similarly, $\Gamma' \setminus e$ is the union of the subgraphs $\perp$ and $R$ that have the unique node $b=c$ in common.
So, $\Gamma' \setminus e = (\perp \sqcup R)/(b,c)$.
Hence $ \mC_{\Gamma' \setminus e}=\mC_{\perp} \oplus \mC_R $ by Lemma \ref{p-sum-code-graph-pol}.
But~$\perp$ is the 180 degrees rotated image of $T$, except that the direction of some edges might be inverted by Remark \ref{r-flype-tangle}.
So, $\mC_T\cong \mC_{\perp }$, where the equivalence is modulo multiplication by $\pm 1 $ of coordinates. 
That is $\mC_T = x* \mC_{\perp }$, where $x$ is a vector with entries $\pm 1 $, and $*$ denotes the coordinate-wise multiplication.
Hence, $ \mC_{\Gamma' \setminus e} = (x,y)* \mC_{\Gamma \setminus e}$, where $y$ is the all-ones vector of length the number of edges of $R$.

Contracting the edge $e$ in $\Gamma $ and $\Gamma'$ gives the graphs $\Gamma / e$ and $\Gamma' / e$, respectively, 
as shown in Figure \ref{fig:flype4}.

\begin{figure}[htbp]
\centering
\begin{subfigure}{0.49\textwidth}
\centering
\resizebox{1.1\textwidth}{!}{
\begin{tikzpicture}
\draw[thick] (0,1) rectangle (1.5,-1);
    \node at (0.6,-0.1) {T};

    \draw[thick] (5,1) rectangle (6.5,-1);
    \node at (5.9,-0.1) {R};

    \draw[black, thick]
    (0.15,0.8)  
    .. controls (-3,-3) and (1,-4) .. (3,-3);
    \draw[black, thick]
    (0.15,0.4)  
    .. controls (-2.5,-2.5) and (1,-3.5) .. (3,-3);
    \draw[black, thick]
    (0.15,-0.5)  
    .. controls (-1.2,-1) and (1,-2.5) .. (3,-3);
    \draw[black, thick]
    (3,-3)  
    .. controls (5,-2.5) and (7.6,-1.1) .. (6.35,-0.5);
    \draw[black, thick]
    (3,-3)  
    .. controls (5,-3.5) and (8.5,-2.5) .. (6.35,0.4);
    \draw[black, thick]
    (3,-3)  
    .. controls (5,-4) and (9,-3) .. (6.35,0.8);

    \draw[black, thick]
    (0.9,0.8) .. controls (1.2,0.6) .. (3.2,0);
    \draw[black, thick]
    (0.9,0.4) .. controls (1.2,0.2) .. (3.2,0);
    \draw[black, thick]
    (0.9,-0.5) .. controls (1.2,-0.3) .. (3.2,0);

    \draw[black, thick]
    (5.4,0.8) .. controls (5,0.6) .. (3.2,0);
    \draw[black, thick]
    (5.4,0.4) .. controls (5,0.2) .. (3.2,0);
    \draw[black, thick]
    (5.4,-0.5) .. controls (5,-0.3) .. (3.2,0);

    \node at (5.4,0) {$\vdots$};
    \node at (1,0) {$\vdots$};
    \node at (5,-2.5) {$\vdots$};
    \node at (1,-2.5) {$\vdots$};
    \node at (6.35,0) {$\vdots$};
    \node at (0.15,0) {$\vdots$};

    \node[above] at (3,-2.8) {$a=d$};
    \node[above] at (3,0.2) {$b=c$};
\end{tikzpicture}
}
\caption{Graph $\Gamma / e$.}
\label{fig:flype4_v1}
\end{subfigure}
\hfill 
\begin{subfigure}{0.49\textwidth}
\centering
\resizebox{1.1\textwidth}{!}{
\begin{tikzpicture}
\draw[thick] (0,1) rectangle (1.5,-1);
    \node at (0.6,-0.1) {\rotatebox{180}{T}};

    \draw[thick] (5,1) rectangle (6.5,-1);
    \node at (5.9,-0.1) {R};

    \draw[black, thick]
    (0.15,0.8)  
    .. controls (-3,-3) and (1,-4) .. (3,-3);
    \draw[black, thick]
    (0.15,0.4)  
    .. controls (-2.5,-2.5) and (1,-3.5) .. (3,-3);
    \draw[black, thick]
    (0.15,-0.5)  
    .. controls (-1.2,-1) and (1,-2.5) .. (3,-3);
    \draw[black, thick]
    (3,-3)  
    .. controls (5,-2.5) and (7.6,-1.1) .. (6.35,-0.5);
    \draw[black, thick]
    (3,-3)  
    .. controls (5,-3.5) and (8.5,-2.5) .. (6.35,0.4);
    \draw[black, thick]
    (3,-3)  
    .. controls (5,-4) and (9,-3) .. (6.35,0.8);
    
    \draw[black, thick]
    (0.9,0.8) .. controls (1.2,0.6) .. (3.2,0);
    \draw[black, thick]
    (0.9,0.4) .. controls (1.2,0.2) .. (3.2,0);
    \draw[black, thick]
    (0.9,-0.5) .. controls (1.2,-0.3) .. (3.2,0);

    \draw[black, thick]
    (5.4,0.8) .. controls (5,0.6) .. (3.2,0);
    \draw[black, thick]
    (5.4,0.4) .. controls (5,0.2) .. (3.2,0);
    \draw[black, thick]
    (5.4,-0.5) .. controls (5,-0.3) .. (3.2,0);

    \node[above] at (3,-2.8) {$a=d$};
    \node at (5.4,0) {$\vdots$};
    \node at (1,0) {$\vdots$};
    
    \node at (6.35,0) {$\vdots$};
    \node at (0.15,0) {$\vdots$};
    \node[above] at (3,0.2) {$b=c$};

\end{tikzpicture}
}
\caption{Graph $\Gamma' / e$.}
\label{fig:flype4_v2}
\end{subfigure}
 \caption{Graphs $\Gamma /e$ and $\Gamma' / e$.}
 \label{fig:flype4}
\end{figure}
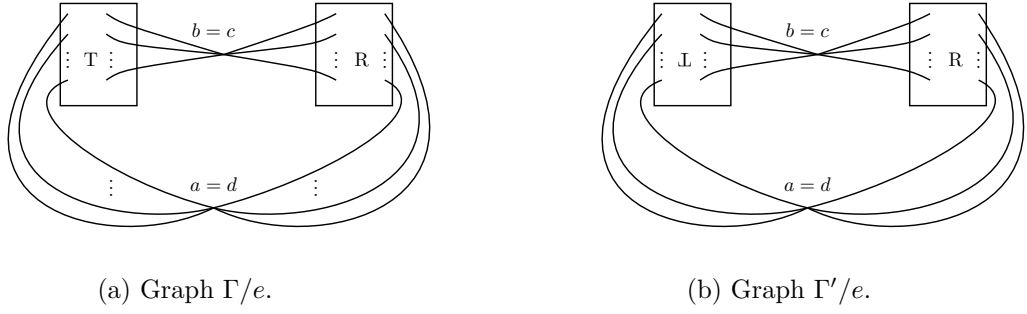

For the contracted graphs we have 

$$\Gamma / e=(T\sqcup R)/((a,b),(d,c)) \mbox{ and } \Gamma' / e=(\perp \sqcup R)/((a,b),(d,c)),$$

see Figure \ref{fig:flype4}. They are equal, after identifying the vertices and edges of $\perp$ and $T$, modulo inverting the direction of edges by Remark \ref{r-flype-tangle}.
Hence~$\mC_{\Gamma' / e} =(x,y)* \mC_{\Gamma / e}$. 

Let $i$ be the position in the cycle code $\mC_{\Gamma}$ corresponding to the edge $e$ of the graph~$\Gamma$.
Then $\mC_{\Gamma \setminus e}= (\mC_{\Gamma})^i$ and $\mC_{\Gamma / e}= (\mC_{\Gamma})_i$  by Proposition \ref{p-contr-del-punct-short}.
Similarly, $\mC_{\Gamma' \setminus e}= (\mC_{\Gamma'})^i$ and~$\mC_{\Gamma' / e}= (\mC_{\Gamma'})_i$.
Hence $(\mC_{\Gamma})^i= ((x,y)*\mC_{\Gamma'})^i$ and $(\mC_{\Gamma})_i= ((x,y)*\mC_{\Gamma'})_i$
Therefore~$\mC_{\Gamma}= (x,y)*\mC_{\Gamma'}$ by Proposition \ref{p-code-punct-short}. Hence  $\mC_{\Gamma}$ and $\mC_{\Gamma'}$ are ($\pm1$)-permutation equivalent.

The proof works only for one of the two black/white diagrams, 
but Theorem \ref{p-dual-black-white} on duality remedies this. \qedhere
\end{proof}

Theorem \ref{thm-code-alternating} is shown for $t=-1$, and this assumption is used in the proof using the property that the black and white codes are dual to each other. In fact the theorem is not true for arbitrary invertible $t$ as we show in the next example.

\begin{figure}[hbt!]
    \centering
\subcaptionbox{$H_t = \begin{pmatrix}
0 & 0 & t & 1 & t & 1 & 0  \\
0 & 0 & 0 & 0 & 0 & t & 1  \\
1 & 0 & 0 & 0 & 1 & 0 & t  \\
t & t & 0 & 0 & 0 & 0 & 0  \\
0 & 1 & 1 & t & 0 & 0 & 0  
\end{pmatrix}$}
[.48\linewidth]{\includegraphics[width=0.7\textwidth]{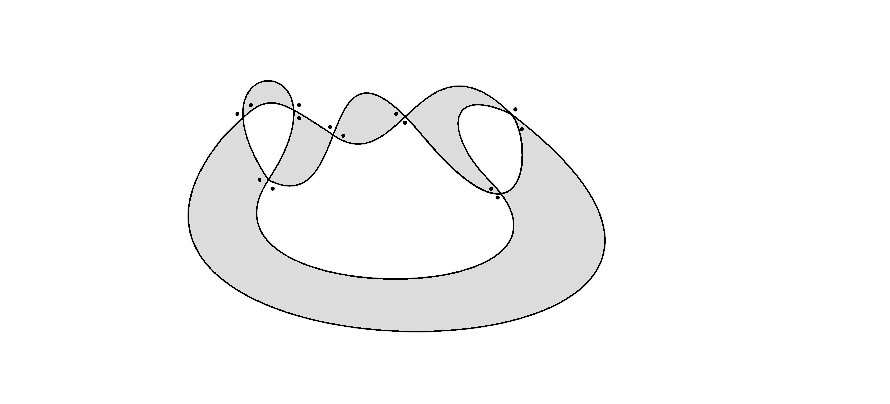}}
 \hspace{.02\textwidth}
\subcaptionbox{$H'_t = \begin{pmatrix}
1 & 0 & t & 1 & 0 & 0 & 0  \\
0 & 0 & 0 & 0 & 0 & t & 1  \\
t & 0 & 0 & 0 & t & 1 & 0   \\
0 & t & 0 & 0 & 1 & 0 & t  \\
0 & 1 & 1 & t & 0 & 0 & 0  
\end{pmatrix}$}
[.48\linewidth]{\includegraphics[width=0.7\textwidth]{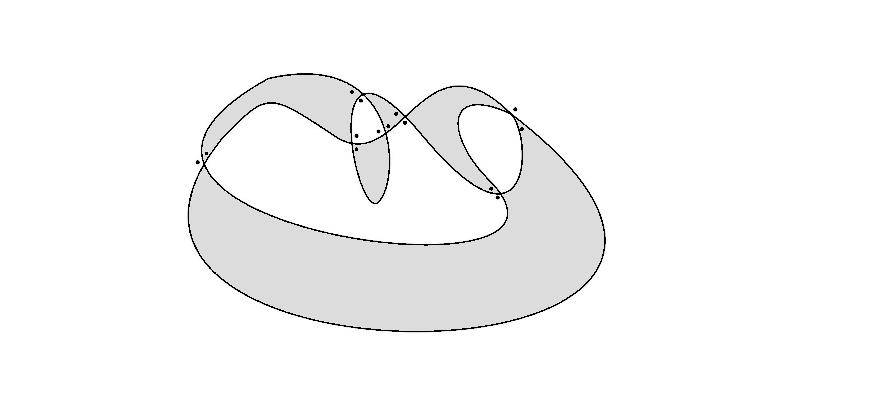}}
\caption{Alternating diagram and a flype of the $7_6$ knot.} \label{fig:flype-6}
\end{figure}


\begin{example}
Consider the alternating Tait diagram $D$ and a flype $D_f$ of the oriented knot~$K=7_6$ as given in Figure \ref{fig:flype-6}. Consider the black graph code of $D$ with $t=2$ over $\F_5$.
Then a generator matrix of the code $\mC_{D,2}$ is given by 
$$
G = \begin{pmatrix}
4 & 1 & 4 & 0 & 1 & 0 & 0\\
4 & 1 & 3 & 3 & 0 & 1 & 3
\end{pmatrix}
$$
and a generator matrix of the code $\mC_{D_f,2}$ is given by  
$$
G_f = \begin{pmatrix}
4 & 2 & 3 & 0 & 1 & 0 & 0\\
2 & 2 & 1 & 1 & 0 & 1 & 4
\end{pmatrix}.
$$
The two codes are not monomial equivalent, since monomial equivalent generator matrices are given by
$$
G' = \begin{pmatrix}
1 & 1 & 1 & 0 & 1 & 0 & 0\\
1 & 1 & 2 & 1 & 0 & 1 & 1
\end{pmatrix}
\ \mbox{ and } \ 
G'_f = \begin{pmatrix}
1 & 1 & 1 & 0 & 1 & 0 & 0\\
3 & 1 & 2 & 1 & 0 & 1 & 1
\end{pmatrix}.
$$
The Alexander polynomial of $K$ is $\Delta_K (t)=t^4-5t^3+7t^2-5t+1$. 
Then $\Delta_K (2)=-5$. So the knot has a nontrivial $(\F_5.2)$-coloring. The corresponding AB-code over $\F_5$ is generated by $(3,2,2,3,1,1,3)$ for the diagram $D$ and by $(1,4,4,1,1,1,3)$ for $D_f$. Hence these codes are not permutation equivalent and not ($\pm1$)-permutation equivalent.

\end{example}

We now present the coding-theoretic invariant we were aiming for.

\begin{theorem} [Theorem \ref{cor-code-alternating}]
The ($\pm1$)-permutation equivalence class of the Alexander-Briggs code $\mA_{D}$ of a reduced Tait diagram $D$ of a prime alternating knot is 
an invariant of the knot.
\end{theorem}

\begin{proof}
The Alexander-Briggs code $\mA_{D}$ is equal to the hull of the cycle code $\mC_{\Gamma _D}$ by Proposition \ref{p-AB-hull}. 
The ($\pm1$)-permutation equivalence class of the code $\mC_{\Gamma _D}$ is an invariant of the knot by Theorem \ref{thm-code-alternating}.
The hull of a code is invariant under the ($\pm1$)-permutation equivalence by Remark \ref{rem:hullofcodes}.
This proves the theorem.
\end{proof}

It has to be noted that deciding whether two reduced alternating knot diagrams represent the same knot is a polynomial-time problem. This was established recently in \cite{haider2024polynomial}, where they proved that given two alternating diagrams, deciding whether they represent the same link/knot admits a polynomial-time algorithm. It is worth noting that the complexity of the equivalence problem remains unknown for arbitrary links. 

\begin{remark}
\label{rem:extendtolinks}
Our results and techniques also extend to links as the Tait flyping conjecture holds for links. However, for the ease of exposition, in this paper we concentrate on knots.
\end{remark}

We conclude this subsection with a remark discussing role of the primeness assumption in the statement of Tait's flyping conjecture.   

\begin{remark}
\label{removeprimeness}
In the original proof of the Tait’s flyping conjecture in \cite{menasco1991-1} and \cite{menasco1993} it is explicitly assumed that the knots are prime, but in several sources it is stated without the ``prime'' assumption, but also without a corresponding proof, see for instance \cite[Chapter 11 \S 5]{murasugi1996knot} and \cite{menasco2021}. However, there does not seem to be any counterexample in the case of an alternating knot that is not prime. That is, there is no known pair of composite alternating knots that are equivalent but not flype-equivalent. We keep the prime assumptions in Theorem \ref{thm-code-alternating} and Theorem \ref{cor-code-alternating}.

\end{remark}

\section{Comparing the New Invariant}
\label{sec:comparenewinv}

In this section, we compare our invariant with several classical ones. We present examples in which our invariant outperforms some well-known invariants in knot theory.

We start with a remark that is implied by Proposition \ref{p-tutte} and the fact that the extended weight enumerator of a graph code and the Tutte polynomial of its graph determine each other. See \cite[\S 8.2.5]{pellikaan2018} and \cite{jurrius2012, jurrius2013}.

\begin{remark}
The extended weight enumerator of the code of a reduced Tait diagram of a prime alternating knot is an invariant. 
\end{remark}

We continue by noting that the Tutte polynomial determines the Jones polynomial, as proven in \cite[page 232]{kauffman1988a}. The converse is not valid, see \cite{thistlethwaite1987}. Therefore, knots having the same Jones polynomial can have distinct Tutte polynomials, and their graph codes will not be equivalent by Proposition \ref{p-tutte}.

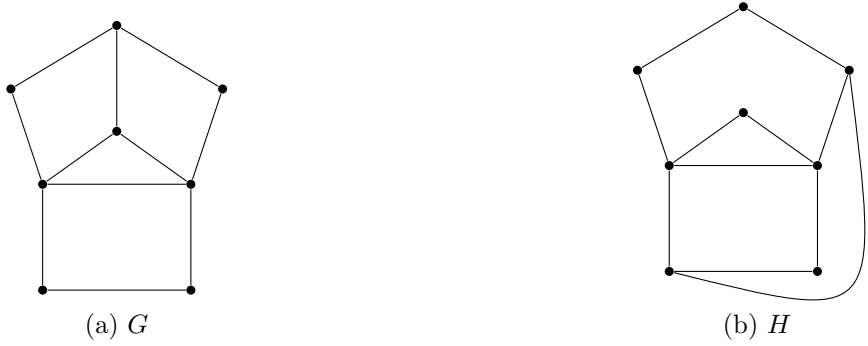
\begin{figure}[htbp]
  \centering

  \begin{subfigure}[b]{0.45\textwidth}
    \centering
    \begin{tikzpicture}[scale=1.4, every node/.style={circle, fill=black, inner sep=1.2pt}]
      \node (A) at (0,2.5) {};   
      \node (B) at (-1,1.9) {};  
      \node (C) at (1,1.9) {};   
      \node (D) at (-0.7,1) {};  
      \node (E) at (0.7,1) {};   
      \node (F) at (-0.7,0) {};  
      \node (G) at (0.7,0) {};   
      \node (H) at (0,1.5) {};   

      \draw (A)--(B)--(D)--(F)--(G)--(E)--(C)--(A);
      \draw (D)--(H)--(E)--(D);
      \draw (A)--(H);
    \end{tikzpicture}
    \caption{$G$}
    \label{fig:G3}
  \end{subfigure}
  \hfill
  \begin{subfigure}[b]{0.45\textwidth}
    \centering
    \begin{tikzpicture}[scale=1.4, every node/.style={circle, fill=black, inner sep=1.2pt}]
      \node (A) at (0,2.5) {};   
      \node (B) at (-1,1.9) {};  
      \node (C) at (1,1.9) {};   
      \node (D) at (-0.7,1) {};  
      \node (E) at (0.7,1) {};   
      \node (F) at (-0.7,0) {};  
      \node (G) at (0.7,0) {};   
      \node (H) at (0,1.5) {};   

      \draw (A)--(B)--(D)--(F)--(G)--(E)--(C)--(A);
      \draw (D)--(H)--(E)--(D);
      \draw (F) .. controls (1.3,-0.5) .. (C);
    \end{tikzpicture}

    \vspace{-0.4cm}

    \caption{$H$}
    \label{fig:G4}
  \end{subfigure}

  \caption{Graphs $G$ and $H$ have the same Tutte polynomial.}
  \label{fig:sametutte}
\end{figure}

Next, we give an example where the Tutte polynomials are the same, but the Alexander polynomial and the codes can still distinguish between them.

\begin{example}
In \cite[Example 3.2]{clancy2015}, two graphs with the Tutte polynomial 
\begin{multline*}
x^7 + 4x^6 + x^5 y + 9x^5 + 6x^4 y + 3x^3 y^2 + x^2 y^3 + 13x^4 + 13x^3 y + 7x^2 y^2 \\
+ 3x y^3 + y^4 + 12x^3 + 15x^2 y + 9x y^2 + 3y^3 + 7x^2 + 9x y + 4y^2 + 2x + 2y
\end{multline*}
are given. We show them in Figure \ref{fig:sametutte}. Since they have the same Tutte polynomial, their Jones polynomial are the same.

Starting from these graphs, we can create alternating knot diagrams by putting crossings in the middle of each edge of the graph (alternatively, you can think of the graph vertices as being in the middle of the regions of the corresponding knot diagram we construct). By using Dowker-Thistlethwaite (DT) notation, we can identify these knots in \cite{knotinfo}. Note that Dowker-Thistlethwaite codes for prime knots determine the knot up to its mirror image \cite{dowker1983classification}. The DT notations of our diagrams are
$$16\ 10\ 22\ 20\ 2\ 14\ 18\ 8\ 4\ 12\ 6 \mbox{, and}$$
$$-14\ -18\ -20\ -12\ -22\ -16\ -2\ -8\ -6\ -4\ -10 \mbox{, respectively.}$$
Again using \cite{knotinfo}, we can determine the knots as $11a_{317}$ and $11a_{67}$, respectively, following the Hoste-Thistlethwaite-Weeks (HTW) notation \cite{hoste1998first}. One can now double-check that the black diagram (see Definition \ref{def:signedgraph}) of $11a_{317}$ gives $G$ of Figure \ref{fig:G3} and the black diagram of $11a_{67}$ gives $H$ of Figure \ref{fig:G4}. 

We can already note at this stage that their Alexander polynomials are distinct, but their determinants are both 125. Furthermore, the second elementary ideals (see Definition \ref{d-elem-id}) are not the same: $(5,t+1)$ and $(1)$, respectively.
So the AB codes, which are the hulls of the graph codes, have dimension 2 and 1, respectively, over $\F_5$. This implies that their corresponding graph codes over $\F_5$ are not equivalent since if two links/knots have the same graph code, then also the same hull (but not necessarily vice versa, as the hull is trivial, in general, unless the characteristic of the field divides the determinant). 
\end{example}

We conclude this section with an example of two alternating links (rather than knots, following Remark \ref{rem:extendtolinks}) with the same Tutte, Alexander and Jones polynomial, where our invariant distinguishes them.

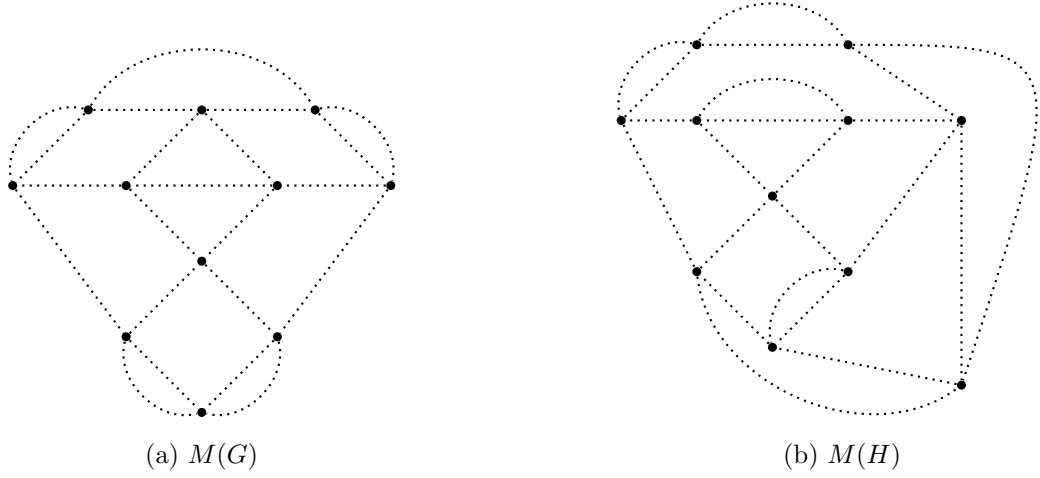
\begin{figure}[htbp]
  \centering

  \begin{subfigure}[b]{0.45\textwidth}
    \centering
    \begin{tikzpicture}[scale=2, every node/.style={circle, fill=black, inner sep=1.2pt}]
      
      \node (A1) at (-0.75,1) {};
      \node (A2) at (0,1) {};
      \node (A3) at (0.75,1) {};
      \node (A4) at (-1.25,0.5) {};
      \node (A5) at (-0.5,0.5) {};
      \node (A6) at (0.5,0.5) {};
      \node (A7) at (1.25,0.5) {};
      \node (A8) at (0,0) {};
      \node (A9) at (-0.5,-0.5) {};
      \node (A10) at (0.5,-0.5) {};
      \node (A11) at (0,-1) {};
      
\draw[dotted,thick] (A1)--(A2)--(A3)--(A7)--(A6)--(A5)--(A4)--
(A9)--(A8)--(A10)--(A11);
\draw[dotted,thick] (A1)--(A4);
\draw[dotted,thick] (A2)--(A5);
\draw[dotted,thick] (A2)--(A6);
\draw[dotted,thick] (A5)--(A8);
\draw[dotted,thick] (A6)--(A8);
\draw[dotted,thick] (A7)--(A10);
\draw[dotted,thick] (A9)--(A11);
\draw[dotted,thick, bend left=60] (A1) to (A3);
\draw[dotted,thick, bend right=60] (A1) to (A4);
\draw[dotted,thick, bend left=60] (A3) to (A7);
\draw[dotted,thick, bend right=60] (A9) to (A11);
\draw[dotted,thick, bend left=60] (A10) to (A11);
      
    \end{tikzpicture}
    \caption{$M(G)$}
    \label{fig:MG}
  \end{subfigure}
  \hfill
  \begin{subfigure}[b]{0.45\textwidth}
    \centering
     \begin{tikzpicture}[scale=2, every node/.style={circle, fill=black, inner sep=1.2pt}]
      \node (A1) at (-0.5,1) {};
      \node (A2) at (0.5,1) {};
      \node (A11) at (1.25,-1.25) {};
      \node (A3) at (-1,0.5) {};
      \node (A4) at (-0.5,0.5) {};
      \node (A5) at (0.5,0.5) {};
      \node (A6) at (1.25,0.5) {};
      \node (A7) at (0,0) {};
      \node (A8) at (-0.5,-0.5) {};
      \node (A9) at (0.5,-0.5) {};
      \node (A10) at (0,-1) {};
      
\draw[dotted,thick] (A3)--(A1)--(A2)--(A6)--(A11)--(A10)--(A8)--
(A7)--(A5)--(A4)--(A3);
\draw[dotted,thick] (A4)--(A7);
\draw[dotted,thick] (A5)--(A6);
\draw[dotted,thick] (A7)--(A9);
\draw[dotted,thick] (A9)--(A10);
\draw[dotted,thick] (A3)--(A8);
\draw[dotted,thick] (A6)--(A9);
\draw[dotted,thick, bend left=60] (A1) to (A2);
\draw[dotted,thick, bend right=60] (A1) to (A3);
\draw[dotted,thick, bend left=60] (A4) to (A5);
\draw[dotted,thick, bend right=60] (A9) to (A10);
\draw[dotted,thick, bend right=60] (A8) to (A11);

\draw[dotted,thick] (A2) .. controls (2,1) .. (A11);
    \end{tikzpicture}
   
    \vspace{-0.4cm}

    \caption{$M(H)$}
    \label{fig:MH}
  \end{subfigure}

  \caption{Medial graphs of $G$ and $H$ in Figure \ref{fig:sametutte}.}
  \label{fig:medialgraphs}
\end{figure}

\begin{example}
Consider the graphs $G$ and $H$ of Figure \ref{fig:sametutte}. These graphs have the same Tutte polynomial. We now follow Jaeger's method in \cite{jaeger1988tutte} which would give two alternating links $\mA(G)$ and $\mA(H)$ whose HOMFLY polynomials \cite{lickorish1987polynomial} are the same. Recall that the HOMFLY polynomial generalizes both Alexander and Jones polynomial. The first step in Jaeger's method is to construct the medial graphs of $G$ and $H$, see Figure \ref{fig:medialgraphs}. Note that the constructed medial graphs are 4-regular. The next step is to checkerboard color the regions and give a direction to every edge such that the incident black region always stays on the left, see Figure \ref{fig:orientedMGMH}. Observe that for every vertex, there are 2 incoming edges and 2 outgoing edges. See Figure \ref{fig:orientedMGMH}. Now the key step is to replace the neigbourhood of each vertex of the medial graph by a diagram of the type depicted in \cite[Figure 3, page 649]{jaeger1988tutte}. Very informally speaking, each edge is replaced by a pair of crossing in a way that the direction of the edges given in Figure \ref{fig:orientedMGMH} is kept the same. This method produces an oriented link diagram whose number of crossing is equal to twice the number of edges of the original graph that we initially start with. 
\end{example}

\begin{figure}[htbp]
  \centering

  \begin{subfigure}[b]{0.45\textwidth}
    \centering
    \begin{tikzpicture}[scale=2, every node/.style={circle, fill=black, inner sep=1.2pt}]
      \coordinate (A1) at (-0.75,1) {};
      \coordinate (A2) at (0,1) {};
      \coordinate (A3) at (0.75,1) {};
      \coordinate (A4) at (-1.25,0.5) {};
      \coordinate (A5) at (-0.5,0.5) {};
      \coordinate (A6) at (0.5,0.5) {};
      \coordinate (A7) at (1.25,0.5) {};
      \coordinate (A8) at (0,0) {};
      \coordinate (A9) at (-0.5,-0.5) {};
      \coordinate (A10) at (0.5,-0.5) {};
      \coordinate (A11) at (0,-1) {};
      


\fill[gray!30] 
        (A9) -- (A8) -- (A5) -- (A4) -- (A9);
\fill[gray!30] 
        (A6) -- (A8) -- (A10) -- (A7) -- (A6);
\fill[gray!30] 
        (A2) -- (A5) -- (A6) -- (A2);

\fill[gray!30] (A1) -- (A3) to[bend right=60] (A1) ;
\fill[gray!30] (A4) -- (A1) to[bend right=60] (A4) ;
\fill[gray!30] (A3) -- (A7) to[bend right=60] (A3) ;
\fill[gray!30] (A11) -- (A9) to[bend right=60] (A11) ;
\fill[gray!30] (A10) -- (A11) to[bend right=60] (A10) ;

\draw[dotted,thick, postaction={decorate},
      decoration={markings, mark=at position 0.5 with {\arrow{>}}}] (A1)--(A2);

\draw[dotted,thick, postaction={decorate},
      decoration={markings, mark=at position 0.5 with {\arrow{>}}}] (A2)--(A3);

\draw[dotted,thick, postaction={decorate},
      decoration={markings, mark=at position 0.5 with {\arrow{>}}}] (A3)--(A7);

\draw[dotted,thick, postaction={decorate},
      decoration={markings, mark=at position 0.5 with {\arrow{>}}}] (A7)--(A6);

\draw[dotted,thick, postaction={decorate},
      decoration={markings, mark=at position 0.5 with {\arrow{>}}}] (A5)--(A6);

\draw[dotted,thick, postaction={decorate},
      decoration={markings, mark=at position 0.5 with {\arrow{>}}}] (A5)--(A4);

\draw[dotted,thick, postaction={decorate},
      decoration={markings, mark=at position 0.5 with {\arrow{>}}}] (A4)--(A9);

\draw[dotted,thick, postaction={decorate},
      decoration={markings, mark=at position 0.5 with {\arrow{>}}}] (A9)--(A8);

\draw[dotted,thick, postaction={decorate},
      decoration={markings, mark=at position 0.5 with {\arrow{>}}}] (A8)--(A10);

\draw[dotted,thick, postaction={decorate},
      decoration={markings, mark=at position 0.5 with {\arrow{>}}}] (A10)--(A11);

\draw[dotted,thick, postaction={decorate},
      decoration={markings, mark=at position 0.5 with {\arrow{>}}}] (A4)--(A1);
\draw[dotted,thick, postaction={decorate},
      decoration={markings, mark=at position 0.5 with {\arrow{>}}}] (A2)--(A5);
\draw[dotted,thick, postaction={decorate},
      decoration={markings, mark=at position 0.5 with {\arrow{>}}}] (A6)--(A2);
\draw[dotted,thick, postaction={decorate},
      decoration={markings, mark=at position 0.5 with {\arrow{>}}}] (A8)--(A5);
\draw[dotted,thick, postaction={decorate},
      decoration={markings, mark=at position 0.5 with {\arrow{>}}}] (A6)--(A8);
\draw[dotted,thick, postaction={decorate},
      decoration={markings, mark=at position 0.5 with {\arrow{>}}}] (A10)--(A7);
\draw[dotted,thick, postaction={decorate},
      decoration={markings, mark=at position 0.5 with {\arrow{>}}}] (A11)--(A9);
\draw[dotted,thick, bend right=60, postaction={decorate},
      decoration={markings, mark=at position 0.5 with {\arrow{>}}}] (A3) to (A1);
\draw[dotted,thick, bend right=60, postaction={decorate},
      decoration={markings, mark=at position 0.5 with {\arrow{>}}}] (A1) to (A4);
\draw[dotted,thick, bend right=60, postaction={decorate},
      decoration={markings, mark=at position 0.5 with {\arrow{>}}}] (A7) to (A3);
\draw[dotted,thick, bend right=60, postaction={decorate},
      decoration={markings, mark=at position 0.5 with {\arrow{>}}}] (A9) to (A11);
\draw[dotted,thick, bend right=60, postaction={decorate},
      decoration={markings, mark=at position 0.5 with {\arrow{>}}}] (A11) to (A10);

\foreach \v in {A1,A2,A3,A4,A5,A6,A7,A8,A9,A10,A11}
  \fill[black] (\v) circle (0.75pt);

    \end{tikzpicture}
    \caption{Checkerboard colored, oriented $M(G)$}
    \label{fig:orientedMG}
  \end{subfigure}
  \hfill
  \begin{subfigure}[b]{0.45\textwidth}
    \centering
     \begin{tikzpicture}[scale=2, every node/.style={circle, fill=black, inner sep=1.2pt}]
      \coordinate (A1) at (-0.5,1) {};
      \coordinate (A2) at (0.5,1) {};
      \coordinate (A11) at (1.25,-1.25) {};
      \coordinate (A3) at (-1,0.5) {};
      \coordinate (A4) at (-0.5,0.5) {};
      \coordinate (A5) at (0.5,0.5) {};
      \coordinate (A6) at (1.25,0.5) {};
      \coordinate (A7) at (0,0) {};
      \coordinate (A8) at (-0.5,-0.5) {};
      \coordinate (A9) at (0.5,-0.5) {};
      \coordinate (A10) at (0,-1) {};

\fill[gray!30] 
        (A8) -- (A7) -- (A4) -- (A3) -- (A8);
\fill[gray!30] 
        (A9) -- (A6) -- (A5) -- (A7) -- (A9);
        
\fill[gray!30] (A3) -- (A1) to[bend right=60] (A3) ;
\fill[gray!30] (A1) -- (A2) to[bend right=60] (A1) ;
\fill[gray!30] (A4) -- (A5) to[bend right=60] (A4) ;
\fill[gray!30] (A10) -- (A9) to[bend right=60] (A10) ;
\fill[gray!30] (A11) -- (A10) -- (A8) to[bend right=60] (A11) ;

\fill[gray!30] 
  (A2) -- (A6) -- (A11) 
  .. controls (2,1) .. (A2);


\draw[dotted,thick, postaction={decorate},
      decoration={markings, mark=at position 0.5 with {\arrow{>}}}] (A3)--(A1);
\draw[dotted,thick, postaction={decorate},
      decoration={markings, mark=at position 0.5 with {\arrow{>}}}] (A1)--(A2);
\draw[dotted,thick, postaction={decorate},
      decoration={markings, mark=at position 0.5 with {\arrow{>}}}] (A2)--(A6);
\draw[dotted,thick, postaction={decorate},
      decoration={markings, mark=at position 0.5 with {\arrow{>}}}] (A6)--(A11);
\draw[dotted,thick, postaction={decorate},
      decoration={markings, mark=at position 0.5 with {\arrow{>}}}] (A11)--(A10);
\draw[dotted,thick, postaction={decorate},
      decoration={markings, mark=at position 0.5 with {\arrow{>}}}] (A10)--(A8);
\draw[dotted,thick, postaction={decorate},
      decoration={markings, mark=at position 0.5 with {\arrow{>}}}] (A8)--(A7);
\draw[dotted,thick, postaction={decorate},
      decoration={markings, mark=at position 0.5 with {\arrow{>}}}] (A5)--(A7);
\draw[dotted,thick, postaction={decorate},
      decoration={markings, mark=at position 0.5 with {\arrow{>}}}] (A4)--(A5);
\draw[dotted,thick, postaction={decorate},
      decoration={markings, mark=at position 0.5 with {\arrow{>}}}] (A4)--(A3);

\draw[dotted,thick, postaction={decorate},
      decoration={markings, mark=at position 0.5 with {\arrow{>}}}] (A7)--(A4);
\draw[dotted,thick, postaction={decorate},
      decoration={markings, mark=at position 0.5 with {\arrow{>}}}] (A6)--(A5);
\draw[dotted,thick, postaction={decorate},
      decoration={markings, mark=at position 0.5 with {\arrow{>}}}] (A7)--(A9);
\draw[dotted,thick, postaction={decorate},
      decoration={markings, mark=at position 0.5 with {\arrow{>}}}] (A10)--(A9);
\draw[dotted,thick, postaction={decorate},
      decoration={markings, mark=at position 0.5 with {\arrow{>}}}] (A3)--(A8);
\draw[dotted,thick, postaction={decorate},
      decoration={markings, mark=at position 0.5 with {\arrow{>}}}] (A9)--(A6);
\draw[dotted,thick, postaction={decorate},
      decoration={markings, mark=at position 0.5 with {\arrow{>}}}, bend right=60] (A2) to (A1);
\draw[dotted,thick, bend right=60, postaction={decorate},
      decoration={markings, mark=at position 0.5 with {\arrow{>}}}] (A1) to (A3);
\draw[dotted,thick, bend right=60, postaction={decorate},
      decoration={markings, mark=at position 0.5 with {\arrow{>}}}] (A5) to (A4);
\draw[dotted,thick, bend right=60, postaction={decorate},
      decoration={markings, mark=at position 0.5 with {\arrow{>}}}] (A9) to (A10);
\draw[dotted,thick, bend right=60, postaction={decorate},
      decoration={markings, mark=at position 0.5 with {\arrow{>}}}] (A8) to (A11);

\draw[dotted,thick, postaction={decorate},
      decoration={markings, mark=at position 0.5 with {\arrow{>}}}] (A11) .. controls (2,1) .. (A2);

\foreach \v in {A1,A2,A3,A4,A5,A6,A7,A8,A9,A10,A11}
  \fill[black] (\v) circle (0.75pt);

    \end{tikzpicture}
    
    \vspace{-0.4cm}

    \caption{Checkerboard colored, oriented $M(H)$}
    \label{fig:orientedMH}
  \end{subfigure}

  \caption{Checkerboard colored medial graphs}
  \label{fig:orientedMGMH}
\end{figure}
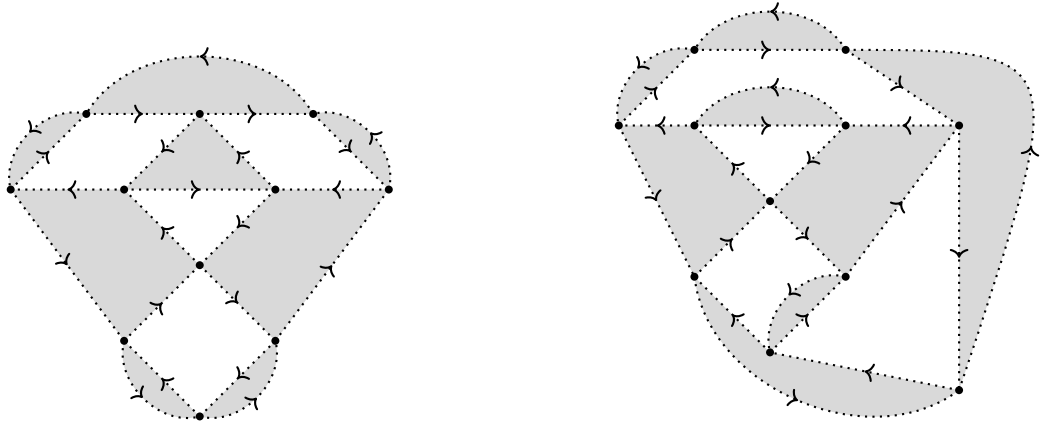

Consider the black graph codes $\mC_{M(G)}$ and $\mC_{M(H)}$ defined over $\F_3$, of the link diagrams obtained by following Jaeger's method in Figure \ref{fig:orientedMGMH}. Then a generator matrix of the code~$\mC_{M(G)}$ is given by

\[
\resizebox{\textwidth}{!}{$
\left[
\begin{tabular}{*{22}{c}} 
0 & 0 & -1 & 1 & 1 & -1 & 1 & -1 & 1 & 1 & 0 & 0 & 0 & 0 & 0 & 0 & 0 & 0 & 0 & 0 & 0 & 0 \\
1 & -1 & 1 & -1 & 0 & 0 & 0 & 0 & 0 & 0 & 1 & 1 & -1 & 1 & 0 & 0 & 0 & 0 & 0 & 0 & 0 & 0 \\
0 & 0 & -1 & 1 & 1 & -1 & 1 & -1 & 0 & 0 & -1 & -1 & 0 & 0 & -1 & 1 & -1 & 1 & -1 & 1 & 0 & 0 \\
0 & 0 & 1 & -1 & -1 & 1 & -1 & 1 & 0 & 0 & 1 & 1 & 0 & 0 & 0 & 0 & 0 & 0 & 0 & 0 & 1 & 1
\end{tabular}
\right]
$}
\]

and a generator matrix of the code $\mC_{M(H)}$ is given by 

\[
\resizebox{\textwidth}{!}{$
\left[
\begin{tabular}{*{22}{c}} 
0 & 0 & -1 & 1 & -1 & 1 & -1 & 1 & -1 & 1 & 0 & 0 & 0 & 0 & 0 & 0 & 0 & 0 & 0 & 0 & 0 & 0 \\
0 & 0 & 1 & -1 & 1 & -1 & 1 & -1 & 0 & 0 & -1 & 1 & -1 & 1 & 0 & 0 & 0 & 0 & 0 & 0 & 0 & 0 \\
0 & 0 & 0 & 0 & -1 & 1 & -1 & 1 & 0 & 0 & 0 & 0 & 0 & 0 & -1 & 1 & -1 & 1 & 0 & 0 & 0 & 0 \\
-1 & 1 & -1 & 1 & -1 & 1 & -1 & 1 & 0 & 0 & 0 & 0 & 0 & 0 & -1 & 1 & 0 & 0 & -1 & 1 & -1 & 1
\end{tabular}
\right]
$}.
\]

The two codes are not monomial equivalent, since their automorphism groups have different orders, which was verified by Magma \cite{bosma1997magma}. The reason monomially equivalent codes must have isomorphic automorphism groups is explained in the following remark.

\begin{remark}    
The automorphism group of a code is the set of all invertible monomial maps that send the code to itself. If two codes $\mC_1$ and $\mC_2$ are monomially equivalent, then there must exist an invertible monomial map $f$ such that $f(\mC_2)=\mC_1$. Now, let $g \in \Aut(\mC_1)$. We then have 
$$f^{-1}gf(\mC_2)= f^{-1}g(\mC_1)=f^{-1}(\mC_1) = \mC_2,$$ implying that $f^{-1}gf \in \Aut(\mC_2)$. So, conjugation by $f$ from $\Aut(\mC_1)$ to $\Aut(\mC_2)$ defines an isomorphism, and therefore these two groups must have the same order. 
\end{remark}

Since the codes are not monomial equivalent, that means they are not ($\pm1$)-permutation equivalent, showing that our invariant can distinguish between these links.

\section{Mutant Knots}
\label{sec:mutant}

In this section we take a closer look at knots that are related by an operation called mutation. They deserve special attention as the knots (the resulting mutants) share many topological properties, as well as sharing many invariants.

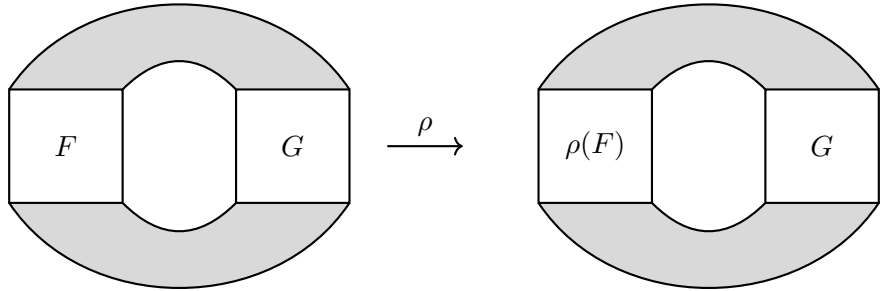
\begin{figure}[htbp]
    \centering
\begin{tikzpicture}

\fill[gray!30]
        (1.5,1.5) 
        .. controls (2,2) and (2.5,2) .. 
        (3,1.5)
        -- (4.5,1.5) 
        .. controls (3.5,3) and (1,3) .. 
        (0,1.5) -- cycle;

    \fill[gray!30]
        (1.5,0) 
        .. controls (2,-0.5) and (2.5,-0.5) .. 
        (3,0) 
        -- (4.5,0) 
        .. controls (3.5,-1.5) and (1,-1.5) .. 
        (0,0) -- cycle;
    
    \draw[thick] (0,1.5) rectangle (1.5,0);
    \node at (0.75,0.75) {$F$};

    \draw[thick] (3,1.5) rectangle (4.5,0);
    \node at (3.75,0.75) {$G$};

    \draw[thick]
        (1.5,1.5) 
        .. controls (2,2) and (2.5,2) .. 
        (3,1.5); 
    \draw[thick]
        (0,1.5) 
        .. controls (1,3) and (3.5,3) .. 
        (4.5,1.5);

    \draw[thick]
        (1.5,0) 
        .. controls (2,-0.5) and (2.5,-0.5) .. 
        (3,0); 
    \draw[thick]
        (0,0) 
        .. controls (1,-1.5) and (3.5,-1.5) .. 
        (4.5,0);

    \draw[->, thick] (5,0.75) -- (6,0.75);

    \fill[gray!30]
        (8.5,1.5) 
        .. controls (9,2) and (9.5,2) ..
        (10,1.5) 
        -- (11.5,1.5) 
        .. controls (10.5,3) and (8,3) ..
        (7,1.5) -- cycle;
    \fill[gray!30]
        (8.5,0) 
        .. controls (9,-0.5) and (9.5,-0.5) .. 
        (10,0) 
        -- (11.5,0) 
        .. controls (10.5,-1.5) and (8,-1.5) ..
        (7,0) -- cycle;

    \draw[thick] (7,1.5) rectangle (8.5,0);
    \node at (7.75,0.75) {$\rho(F)$};

    \draw[thick] (10,1.5) rectangle (11.5,0);
    \node at (10.75,0.75) {$G$};

    \draw[thick]
        (8.5,1.5) 
        .. controls (9,2) and (9.5,2) .. 
        (10,1.5); 
    \draw[thick]
        (7,1.5) 
        .. controls (8,3) and (10.5,3) .. 
        (11.5,1.5);

    \draw[thick]
        (8.5,0) 
        .. controls (9,-0.5) and (9.5,-0.5) .. 
        (10,0); 
    \draw[thick]
        (7,0) 
        .. controls (8,-1.5) and (10.5,-1.5) ..
        (11.5,0); 
    \node[above] at (5.5,0.75) {$\rho$};

\end{tikzpicture}

    \caption{Diagrams $D$ (on the left) and $D_\rho$ (on the right).}
    \label{fig:mutation1}
\end{figure}

Two knots are \textbf{mutants} if one can be obtained from the other by cutting out a tangle (see Definition \ref{def:ntangle}) with four boundary points and ``regluing'' it after a half-turn (180 degrees) around one of the three possible axes (see Figure \ref{fig:mutation_rotation}) in a way that preserves the connections to the rest of the knot at those four boundary points. See Figure \ref{fig:mutation1} for the general picture and see \cite[Figure 3.3, page 29]{lickorish1997} for a nice illustrative example.

\begin{figure}[htbp]
    \centering
    
    \begin{subfigure}{0.3\textwidth}
        \centering
        \begin{tikzpicture}
            \draw[thick] (0,1.5) rectangle (1.5,0);
            \node at (0.75,0.75) {F};
        \draw[thick, ->]
        (1.2,0.75) 
        .. controls (1.2,1.2) and (0.3,1.2) .. 
        (0.3,0.75); 
        \end{tikzpicture}
        \caption{The rotation $\rho_x$.}
        \label{fig:rotation_x}
    \end{subfigure}
    \hfill 
    \begin{subfigure}{0.3\textwidth}
        \centering
        \begin{tikzpicture}
            
            \draw[thick] (0,1.5) rectangle (1.5,0);
            \node at (0.75,0.75) {F};

            \draw[thick] (1.5,0.75) -- (2.6,0.75);

            \draw[thick, ->] (2.6,0.3) .. controls (2.8,1.2) .. (2.35,1);
            
            \draw[thick] (2.8,0.75) -- (3.9,0.75);
        \end{tikzpicture}
        \caption{The rotation $\rho_y$.}
        \label{fig:rotation_y}
    \end{subfigure}
    \hfill 
    \begin{subfigure}{0.3\textwidth}
        \centering

\begin{tikzpicture}
            
            \draw[thick] (0,1.5) rectangle (1.5,0);
            \node at (0.75,0.75) {F};

            \draw[thick] (0.75,0) -- (0.75,-1.1);

            \draw[thick, ->] (0.3,-1.1) .. controls (1.2,-1.3) .. (1,-0.5);

            \draw[thick] (0.75,-1.3) -- (0.75,-2.5);
        \end{tikzpicture}

        \caption{The rotation $\rho_z$.}
        \label{fig:rotation_z}
    \end{subfigure}
    \caption{The three rotations $\rho_x$, $\rho_y$ and $\rho_z$.}
    \label{fig:mutation_rotation}
\end{figure}
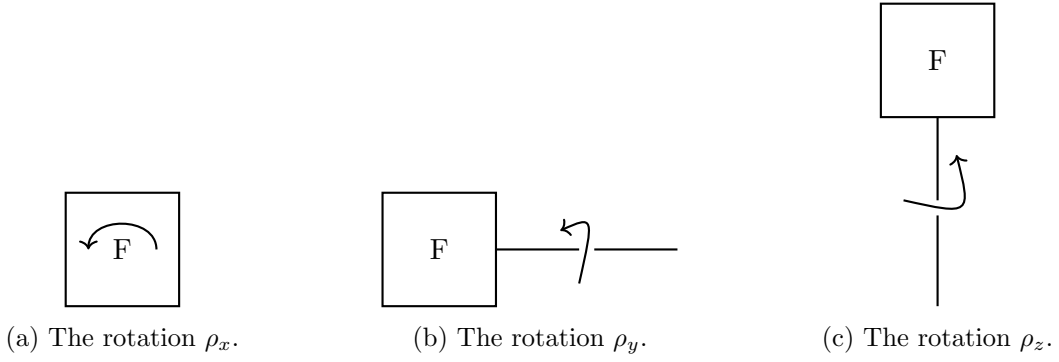

The rotation can be of three different types, corresponding to different orientations, resulting in three mutants. Note that any mutation of a knot is again a knot \cite[Exercise 2.24]{adams1994} and any mutation of an alternating knot is again an alternating knot \cite[Exercise 2.23]{adams1994}, but these three mutants may not be necessarily distinct, that is they may be equivalent knots, see \cite[Figure 2.30 and Figure 2.31]{adams1994}. We denote by $\rho_x$, $\rho_y$, $\rho_z$ the half turn rotations around the $x$-axis, $y$-axis and $z$-axis, respectively. Their products commute, and moreover the following holds:
$$
\rho_x \rho_y = \rho_y \rho_x = \rho_z, \ \ \rho_x \rho_z = \rho_z \rho_x = \rho_y, \ \ \rho_y \rho_z = \rho_z \rho_y = \rho_x.
$$
Furthermore, $\rho_x^2= \rho_y^2=\rho_z^2=1 $ by definition since all three rotations are half turns. As a result, they generate the Klein 4-group
$<\rho_x , \rho_y ,\rho_z>= \{ 1, \rho_x , \rho_y ,\rho_z \}$.

It was shown in \cite[Theorem 3]{lickorish1987b} that
mutants share HOMFLY and Kauffman polynomials, and hence the same Alexander and Jones polynomial. Therefore, new invariants are often considered in their success to distinguish mutant knots \cite{bishler2021}.
The difficulty to distinguish mutant knots is used in a Post Quantum Cryptography scheme \cite{marzuoli2011}. We refer to \cite{friedl2007} for a table is given that lists all mutant pairs of knots with $11$ crossings.

Regarding our new invariant, we conclude the section by showing that the mutants have equivalent codes.

\begin{theorem}\label{thm-mutation}
Let $D$ be a Tait diagram of an alternating knot $K$, and $D_{\rho}$ be the Tait diagram of the $\rho$-mutation of $K$. 
Then the graph codes of $D$ and $D_{\rho }$ are ($\pm1$)-permutation equivalent.  
\end{theorem}

\begin{proof}
We give the proof for the $\rho_z$ mutation. 
We may assume without loss of generality that the top and bottom regions in Figure \ref{fig:mutation1} are black,
since the codes of the black and white graph are dual to each other by Theorem \ref{p-dual-black-white}. \\
The black graphs $\Gamma_D$ and $\Gamma_{D_{\rho_z}}$ in Figure \ref{fig:mutation2_part1} are isomorphic, 
that is they are the same up to possible changes of the direction of the arrows in the graphs $\Gamma_F$ and $\rho_z(\Gamma_F)$, since they
are obtained by rotating over $180$ degree as explained in Remark \ref{r-flype-tangle}. 
So, the graph codes of $D$ and $D_{\rho_z }$ are ($\pm1$)-permutation equivalent.   \\ 
The proof for the $\rho_y$ mutation is analogous, 
and the $\rho_x$ mutation is equal to combination of the the $\rho_y$ and the $\rho_z$ mutations, since the $\rho_x= \rho_y \rho_z$.
Hence the graph codes of $D$ and $D_{\rho }$ are ($\pm1$)-permutation equivalent for all mutations $\rho$.  
\end{proof}


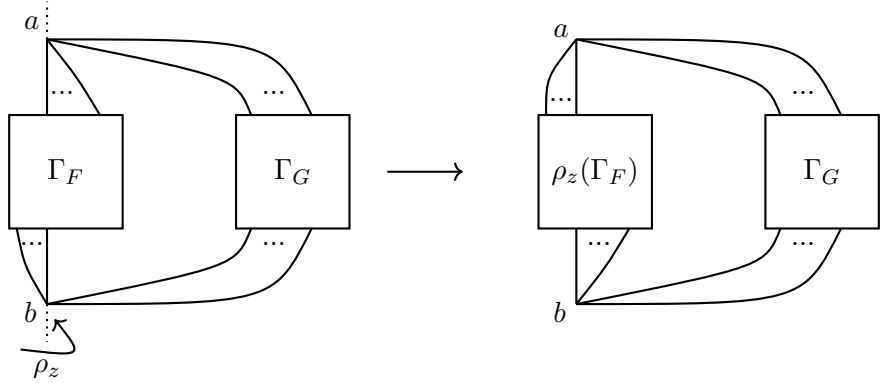
\begin{figure}[htbp]
\centering

\begin{tikzpicture}

    \draw[thick] (0,1.5) rectangle (1.5,0);
    \node at (0.75,0.75) {$\Gamma_F$};

    \draw[thick] (3,1.5) rectangle (4.5,0);
    \node at (3.75,0.75) {$\Gamma_G$};

    \draw[thick, dotted] (0.5,2.5) -- (0.5,3);
    \draw[thick] (0.5,1.5) -- (0.5,2.5);
    \draw[thick] (0.5,0) -- (0.5,-1);
    \draw[thick, dotted] (0.5,-1) -- (0.5,-1.5);
    
    \draw[thick, ->] (0.15,-1.6) .. controls (1,-1.7) .. (0.6,-1.2);
    \node[below] at (0.5,-1.6) {$\rho_z$};

    \node[left] at (0.5,-1.1) {$b$};
    \node[above] at (0.3,2.5) {$a$};
    \draw[thick] (0.5,-1) .. controls (0.2,-0.5) .. (0.1,0);
    \node at (0.3,-0.2) {$...$};
    \draw[thick] (0.5,2.5) .. controls (0.9,2) .. (1.2,1.5);
    \node at (0.7,1.8) {$...$};
    \draw[thick] (0.5,2.5) .. controls (3,2) .. (3.2,1.5);
    \node at (3.5,1.8) {$...$};
    \draw[thick] (0.5,-1) .. controls (3,-0.5) .. (3.2,0);
    \node at (3.5,-0.2) {$...$};

    \draw[thick] (0.5,2.5) .. controls (3.5,2.5) .. (4,1.5);
    \draw[thick] (0.5,-1) .. controls (3.5,-1) .. (4,0);

    \draw[->, thick] (5,0.75) -- (6,0.75);

    \draw[thick] (7,1.5) rectangle (8.5,0);
    \node at (7.75,0.75) {$\rho_z(\Gamma_F)$};

    \node[left] at (7.5,-1.1) {$b$};
    \node[above] at (7.3,2.4) {$a$};
    \draw[thick] (7.5,-1) .. controls (7.9,-0.5) .. (8.2,0);
    \node at (7.8,-0.2) {$...$};
    \draw[thick] (7.5,2.5) .. controls (7.1,2) .. (7.1,1.5);
    \node at (7.3,1.7) {$...$};
    \draw[thick] (7.5,2.5) .. controls (10,2) .. (10.2,1.5);
    \node at (10.5,1.8) {$...$};
    \draw[thick] (7.5,-1) .. controls (10,-0.5) .. (10.2,0);
    \node at (10.5,-0.2) {$...$};

    \draw[thick] (10,1.5) rectangle (11.5,0);
    \node at (10.75,0.75) {$\Gamma_G$};

    \draw[thick] (7.5,1.5) -- (7.5,2.5);
    \draw[thick] (7.5,0) -- (7.5,-1);

    \draw[thick] (7.5,2.5) .. controls (10.5,2.5) .. (11,1.5);
    \draw[thick] (7.5,-1) .. controls (10.5,-1) .. (11,0);   
\end{tikzpicture}
\caption{Black graph $\Gamma_D$ (on the left) and $\Gamma_{D_{\rho_z}}$ (on the right).}
    \label{fig:mutation2_part1}
\end{figure}

\section*{Acknowledgements}
Part of this work includes material that first appeared in the doctoral dissertation of the first author, see \cite[Chapter 9]{kilicthesis}.

\bigskip
\bibliographystyle{abbrv}
\bibliography{ADV}

\end{document}